\documentclass[a4paper,11pt]{article} 
\usepackage[a4paper,hmargin=2.8cm,vmargin=3cm]{geometry}
\usepackage[english]{babel}
\usepackage{amsmath,amsthm,amsfonts,amssymb,mathrsfs,bbm}
\usepackage{mathtools}
\usepackage{authblk}
\usepackage{lineno}
\usepackage{slashed}
\usepackage{comment}
\usepackage[dvipsnames]{xcolor}

\definecolor{BeauBlue}{rgb}{0, 0.2, .85}
\definecolor{BeauOrange}{rgb}{.8, .1, 0}
\usepackage{hyperref}
\hypersetup{
	colorlinks = true,
	linkcolor = BeauBlue,
	urlcolor = BeauBlue,
	citecolor = BeauOrange,
}

\newcommand{\ee}{\mathrm{e}}
\newcommand{\ii}{\mathrm{i}}  
\newcommand{\dd}{\mathrm{d}}
\newcommand{\veps}{\varepsilon}

\newcommand{\llangle}{\langle\!\langle}
\newcommand{\rrangle}{\rangle\!\rangle}

\usepackage{stackengine}
\usepackage{dsfont}
\usepackage{tikz}
\usetikzlibrary{shapes.geometric}
\usetikzlibrary { decorations.pathmorphing, decorations.pathreplacing, decorations.shapes,positioning}
\usepackage{float}

\renewenvironment{thebibliography}[1]{
  \begin{oldthebibliography}{#1}
    \setlength{\itemsep}{0em}
    \setlength{\parskip}{0,5em}
}
{
  \end{oldthebibliography}
}

\numberwithin{equation}{section}

\newcommand{\ind}[1]{\mathds{1}_{#1}}
\newcommand{\cw}{\mathcal{W}}

\newcommand{\cv}{\nu}
\newcommand{\cpsi}{\bar{\psi}}

\newtheorem{theorem}{Theorem}[section]
\newtheorem{corollary}[theorem]{Corollary}
\newtheorem{lemma}[theorem]{Lemma}
\newtheorem{proposition}[theorem]{Proposition}
\newtheorem{remark}[theorem]{Remark}
\newtheorem{definition}[theorem]{Definition}

\title{Exponential decay of correlations at high temperature in $H^{2|2n}$ nonlinear sigma models}
\author[1]{Margherita Disertori}
\author[1]{Javier Dur\'an Fern\'andez}
\author[2]{Luca Fresta}
\affil[1]{Hausdorff Center for Mathematics,
University of Bonn, Endenicher Allee 60
53115 Bonn, Germany \\Email: \href{mailto: mdiserto@uni-bonn.de}
{\texttt{mdiserto@uni-bonn.de}},
\href{mailto: jduranfe@uni-bonn.de}
{\texttt{jduranfe@uni-bonn.de}
}}
\affil[2]{Mathematics Department, University Roma Tre,
Largo San Leonardo Murialdo 1, 00146 Roma, Italy
\\Email: \href{mailto:luca.fresta@uniroma3.it}
{\texttt{luca.fresta@uniroma3.it}} }

\date{\today}

\begin{document}

    \maketitle

\begin{abstract}
We consider a family of nonlinear sigma models on $\mathbb{Z}^{d}$ whose target space is the hyperbolic super manifold $H^{2|2n}$, $n >1$, introduced by Crawford as an extension of Zirnbauer's $H^{2|2}$ model for disordered systems.
We prove exponential decay of the two-point correlation function in the high-temperature regime $\beta \leq C n^{-1}$, with $C>0$ a universal constant, for any $n > 1$ and any dimension $d\geq 1$, with mass $\log \beta^{-1}$. We also consider models with long-range interaction and prove fast decay in the same high-temperature regime.
The proof is based on the reduction to a marginal fermionic theory and combines a high-temperature cluster expansion, exact combinatorics and bounds derived via Grassmann norms.
\end{abstract}
 
\tableofcontents
    
\section{Introduction}

Supersymmetric nonlinear sigma models have emerged as a powerful tool to study disordered systems, providing a framework in which spectral and transport properties of random Schr\"odinger operators and random band matrices can be analysed.

Early pioneering work by Wegner \cite{Wegner1979,Wegner1980} demonstrated that the correlation functions of lattice models with hyperbolic symmetry encode the spectral properties of the Anderson metal-insulator transition. 
These ideas were subsequently put on more systematic footing with the supersymmetric method developed by Efetov \cite{Efetov1983,Efetov1997}, which provides a direct mapping from problems involving random operators to statistical mechanical models formulated in terms of supermatrices, whose entries consist of both commuting (bosonic) and anticommuting (Grassmann) variables. In the particular case of observables involving the squared modulus of resolvent matrix elements, the resulting supersymmetric models exhibit an additional hyperbolic symmetry.

The supermatrix models obtained in this way are still highly nontrivial, and it is customary in the physics literature to consider the so-called sigma-model approximation, which restricts the target space to a suitable nonlinear manifold. From a mathematical perspective, a canonical example of this framework is Zirnbauer’s $H^{2|2}$ model \cite{Zirnbauer1991,Zirnbauer1992}, which captures essential features of random band matrices and provides a minimal setting for the rigorous analysis of the Anderson metal-insulator transition \cite{Spencer2012}.

The $H^{2|2}$ model has been extensively studied and is by now well understood.
In dimensions $d \geq 3$ it has been shown to exhibit quasi-diffusive behaviour at large inverse temperature $\beta$ \cite{DisertoriSpencerZirnbauer2010}, see also \cite{SpencerZirnbauer2004} for spontaneous symmetry breaking in a related nonlinear sigma model. Combined with the proof of localisation at small $\beta$ in any dimension \cite{DisertoriSpencer2010}, these results establish the existence of an Anderson-type transition for the model. 
The mathematical understanding of the $H^{2|2}$ model goes, however, beyond this. In $d=1$, localisation has in fact been proven for all values of $\beta$ \cite{DisertoriSpencer2010}, see also \cite{Zirnbauer1991}, while in $d=2$, where neither extended states nor conventional symmetry breaking are expected, generalised Mermin--Wagner theorems have been established \cite{Bauerschmidt2019}.

Following general insights from the work of Wegner \cite{Wegner1979,Wegner1980}, as well as more recent advances in the study of models with hyperbolic symmetries \cite{Seiler2005,Bauerschmidt2019}, it is  a sensible question to broadly understand nonlinear sigma models with a general target space possessing hyperbolic symmetry. 

A natural candidate in this direction is provided by the hyperbolic superspace $H^{2|2n}$, introduced by Crawford \cite{Crawford2021}. Compared to $H^{2|2}$, the superspace $H^{2|2n}$ contains $n-1$ additional pairs of Grassmann coordinates, and can be shown to be closely related to sigma models with spherical superspace \cite{Caracciolo2017}, which similarly involve a surplus of Grassmann variables. The interest in models with additional Grassmann variables is rooted in the long-standing intuitive idea that a pair of Grassmann variables effectively contributes as ``minus two'' bosonic degrees of freedom \cite{Parisi1979,McKane1980}. In this sense, nonlinear sigma models with target space $H^{2|2n}$ and $n \ge 2$ are expected to realize an analytic continuation of $O(N)$ models to negative values of the parameter, specifically $N = 1 - n$, see \cite{Caracciolo2004,Caracciolo2007,Caracciolo2017}.

Finally, as emphasized in \cite{Crawford2021}, the family of target spaces $H^{2|2n}$, much like the classical $O(N)$ models, introduces an additional parameter that allows one to consider a large-$n$ limit. This feature provides a natural setting in which to explore questions of universality as well as harder problems like that of spontaneous symmetry breaking in two dimensions \cite{Polyakov1975}.

\subsection{The model and its probabilistic marginals}

Let us now describe the $H^{2|2n}$ nonlinear sigma model in more detail.
The superspace $H^{2|2n}$ is parametrised by two real variables and $2n$ independent Grassmann variables\footnote{See \cite{Wegner2016} for an introduction to Grassmann variables and supercalculus.}, and consists of $2n+3$-component supervectors of the form
\begin{equation}\label{eq:def-u}
u := (z,x,y,\xi,\eta, \psi)  \;, \qquad
 \psi :=(\bar{\psi}_{1},\psi_{1},\dots,\bar{\psi}_{n-1}, \psi_{n-1}) \;,
\end{equation}
where $x,y \in \mathbb{R}$, $\xi$, $\eta$ and $\big\{\bar{\psi}_{\alpha},\psi_{\alpha}\big\}_{\alpha=1,\dots,n-1}$ are independent Grassmann variables. Finally, we defined
\begin{equation}\label{eq:z}
z= \sqrt{1+x^{2} + y^{2} + 2 \xi \eta + \psi \cdot  \psi} \;, 
\end{equation}
having introduced the symmetric product for Grassmann variables
\begin{equation}\label{eq:symmetric-product}
\psi \cdot \psi' :=\sum_{\alpha = 1}^{n-1}\big(\bar{\psi}_{\alpha} \psi_{\alpha}' + \bar{\psi}_{\alpha}' \psi_{\alpha} \big) \;.
\end{equation}
The supervector $u$ satisfies the hyperbolic sigma model constraint $\langle u ,u \rangle = -1$ with symmetric product
\begin{equation*}
\langle u , u' \rangle := -z z' + x x'+y y'+ \xi \eta'+ \xi '\eta +  \psi \cdot \psi' \;.
\end{equation*}
This symmetric product is indefinite with signature $(-,+,\dots,+)$, so that it is natural to think of the loci $\langle u ,u \rangle = -1$ as generalisation of hyperbolic planes with Grassmann variables.

We then consider a finite box $\Lambda \subset \mathbb{Z}^{d}$ and let $\big(H^{2|2n}\big)^{\Lambda}$ be the configuration space. Elements in  $\big(H^{2|2n}\big)^{\Lambda}$ are of the form $u = (u_{j})_{j\in \Lambda}$, where each $u_{j}$ is as in \eqref{eq:def-u}
\begin{equation*}
u_{j}=(z_{j},x_{j},y_{j}, \xi_{j},\eta_{j},  \psi_{j}) \;, \qquad
\psi_{j}:=(\bar{\psi}_{j,1}, \psi_{j,1}, \dots ,\bar{\psi}_{j,n-1}, \psi_{j,n-1}) \;,
\end{equation*}
where $\xi_{j}$, $\eta_{j}$ and $\{\bar{\psi}_{j,\alpha}, \psi_{j,\alpha} \}_{j \in \Lambda, \alpha =1,\dots,n-1}$  are independent Grassmann variables, for any $j \in \Lambda$.
On the configuration space we introduce the superintegration form
\begin{equation*}
\dd \mu_{\Lambda} := \prod_{j \in \Lambda} \frac{\dd x_{j} \dd y_{j}}{2\pi} \partial_{\xi_{j}} \partial_{\eta_{j}} \Big[ \prod_{\alpha=1}^{n-1}\partial_{\bar{\psi}_{j,\alpha}}\partial_{\psi_{j,\alpha}}\Big]  \frac{1}{z_{j}} \;.
\end{equation*}
Here derivation with respect to Grassmann variables $\partial_{\xi_{j}} \partial_{\eta_{j}} \Big[ \prod_{\alpha =1}^{n-1}\partial_{\bar{\psi}_{j,\alpha}}\partial_{\psi_{j,\alpha}}\Big]$ is simply the standard Berezin integration form, that is, projection over the top coefficient in the Grassmann algebra \cite{Berezin1987}. 

The $H^{2|2n}$ nonlinear sigma model measure on $\Lambda$ at inverse temperature $\beta>0$ and external field (pinning) $\veps > 0$ is defined via the Gibbsian prescription
\begin{equation}\label{eq:Gibbs-measure}
\langle F\rangle^{\Lambda}_{\beta,\veps,n} := \frac{1}{Z_{\beta,\veps,n}^{\Lambda}} \int
 \dd \mu_{\Lambda} \ee^{-H^{\Lambda}_{\beta,\veps}} F \;,
\end{equation}
for any observable $F=F(u)$. Here $Z^{\Lambda}_{\beta,\veps,n}$ is the partition function, that is, the normalisation constant such that $\langle \cdot \rangle^{\Lambda}_{\beta,\veps,n} = 1$, while the Hamiltonian is given by
\begin{equation}\label{eq:def-Hamilt}
\begin{split}
H^{\Lambda}_{\beta,\veps}& := \frac{\beta}{2} \sum_{\{i,j \} \in E_{\Lambda}} J_{ij} \langle u_{i} -u_{j}, u_{i} -u_{j} \rangle + \veps \sum_{j \in \Lambda} (z_{j}-1)
\end{split}
\end{equation}
for some symmetric coupling constants $J_{ij}= J_{ji} \geq 0$, $J_{ii} =0$, where $E_{\Lambda}$ denotes the set of edges in the complete graph on $\Lambda$. 
Note that the diagonal term $i=j$ would anyway not appear because it is a gradient term. The canonical choice is nearest-neighbour interaction, $J_{ij} = \ind{|i-j|=1}$, but the model can be defined for more general, sufficiently-fast-decaying, coefficients $J_{ij}$. 
In our notation, the dependence on $n$ is explicit only through the measure and the normalisation constant, whereas the dependence on $J$ is kept implicit everywhere.

The definition of the $H^{2|2n}$ model in \eqref{eq:Gibbs-measure} is meaningful only if the partition function $Z^{\Lambda}_{\beta,\veps,n}$ is finite and nonzero. For $n=1$, this follows directly by supersymmetry: in fact $Z^{\Lambda}_{\beta,\veps,1}=1$ for any $\beta,\veps >0$, see \cite{DisertoriSpencerZirnbauer2010}. For $n>1$, finiteness and nondegeneracy of the partition function are far from obvious, at least when working in the supervector $u$ coordinates.
 
Nevertheless, after switching to horospherical coordinates \cite{DisertoriSpencerZirnbauer2010} and performing Gaussian integration of the Grassmann variables one finds the following probabilistic representation
\begin{equation}\label{eq:partition-function-VRJP}
Z^{\Lambda}_{\beta,\veps,n} = \int_{\mathbb{R}^{\Lambda}} \prod_{j \in \Lambda}\frac{\dd t_{j}}{(2\pi)^{1/2}} \, \ee^{-F^{\Lambda}_{\beta}(\nabla t)} \ee^{-M_{\veps}^{\Lambda}(t)}\Big[ \det D_{\veps}^{\Lambda}(t)\Big]^{n-\frac{1}{2}} \;.
\end{equation}
Here, we defined
\begin{equation*}
F^{\Lambda}_{\beta}(\nabla t) := \beta \sum_{\{i,j\} \in E_\Lambda} J_{ij} \big( \cosh (t_{i} - t_{j}) - 1\big) \;, 
\qquad \quad
M^{\Lambda}_{\veps}(t):= \veps \sum_{j \in \Lambda} \big( \cosh (t_{j}) - 1\big) \;,
\end{equation*}
and introduced the matrix $D_{\veps}^{\Lambda} = D_{\veps}^{\Lambda}(t) \in \mathbb{R}^{\Lambda \times \Lambda}$
\begin{equation}\label{eq:def-D}
\begin{split}
\big(D_{\veps}^{\Lambda} \big)_{i j} & := - \beta J_{i j} \;, \hspace{9.7em} \forall i \neq j \;, \\
\big(D_{\veps}^{\Lambda} \big)_{j j} & := \beta \sum_{i \in \Lambda} J_{ij} \ee^{t_{i} - t_{j}} + \veps \ee^{-t_{j}} \;, \qquad \quad \forall j  \;.
\end{split}
\end{equation}
Note that $D_{\veps}^{\Lambda} =  -\beta \Delta^{J} + \veps \ee^{-t}$, where $\Delta^{J}$ is the graph Laplacian on $(\Lambda,E_\Lambda)$ with edge weights $J_{ij} $, so that $D_{\veps}^{\Lambda} >0 $ as long as $\veps >0$.  
As a consequence, the integrand in \eqref{eq:partition-function-VRJP} is pointwise positive, and hence $Z^{\Lambda}_{\beta,\veps,n}>0$. Finiteness of $Z^{\Lambda}_{\beta,\veps,n} $ follows by the presence of the term $\ee^{-M_{\veps}^{\Lambda}(t)}$, which ensures integrability for $\veps >0$. It is also worth mentioning that \eqref{eq:partition-function-VRJP}, unlike the original $H^{2|2n}$ model, is well-defined for any $n \in \mathbb{R}$, see \cite{Crawford2021}.

While \eqref{eq:partition-function-VRJP} already provides a probabilistic representation of the partition function, the probabilistic content of the model is considerably more interesting and goes beyond its original physical motivation. In the case $n=1$, a remarkable result of Sabot and Tarr\'es \cite{Sabot2015}, building on previous observations in \cite{DisertoriSpencerZirnbauer2010}, shows that the Gibbs measure related to \eqref{eq:partition-function-VRJP} arises as the mixing measure for the vertex-reinforced jump process (VRJP), model introduced and studied in a completely different context \cite{Davis2002}.
The connection between the $H^{2|2}$ model and the VRJP, thus offers an alternative perspective on the latter: in particular, it allows one to translate the results on localisation and quasi-diffusion for the nonlinear sigma model into recurrence and transience for the VRJP
\cite{Sabot2015}, see also \cite{DisertoriSabot2015}.

Remarkably, also the case $n = 2$ admits a direct connection to a probabilistic model introduced and studied in a different context, namely the arboreal gas. The latter can be viewed as a Bernoulli bond-percolation model with parameter $p = \beta/(1+\beta)$, conditioned to be acyclic\footnote{Technically, this is correct only for $\varepsilon = 0$. On the other hand, adding a nonzero external field $\varepsilon > 0$, as in standard Bernoulli bond percolation, is equivalent to considering the model on an extended graph obtained by adding a distinguished vertex (the graveyard), connected to all vertices with edge weight $\varepsilon$.}, see \cite{Grimett2006,Bauerschmidt2021}.

The connection between the two models can be made explicit as follows. By the supersymmetric localisation theorem, the commuting variables $x,y$ and the Grassmann variables $\xi,\eta$ in the $H^{2|4}$ model can be integrated out exactly, yielding a marginal $H^{0|2}$ model involving only the Grassmann degrees of freedom $\bar\psi$ and $\psi$. The resulting partition function takes the form (see  \eqref{eq:fact-notation} with $m=1$)
\begin{equation}\label{eq:arboreal-Z2}
\int \dd\psi_{\Lambda }
\  \ee^{-\beta \big(\bar{\psi}, (-\Delta^{J} + \veps + 1) \psi \big) -  \frac{\beta}{2}(\bar{\psi}\psi, - \Delta^{J} \bar{\psi}\psi) }\; ,
\end{equation}
where we used the shorthand
$(a, b)\equiv \sum_{i} a_{i} b_{i}$, see Section \ref{sec:expansion} for more details.

The partition function of the arboreal gas can be written as \eqref{eq:arboreal-Z2}, see  \cite{Caracciolo2004,Jacobsen2005,Caracciolo2007}. The proofs in \cite{Caracciolo2004,Jacobsen2005,Caracciolo2007} are purely combinatorial; in Appendix~\ref{sec:arboreal gas}, we present an alternative derivation based on the Hubbard--Stratonovich transformation.

The remarkable connection between the arboreal gas and a purely Grassmann model has recently been used in \cite{Bauerschmidt2024} to prove that the model exhibits a percolation transition at large $\beta$ for any dimension $d \ge 3$, while the work \cite{Bauerschmidt2021} established that, in contrast to Bernoulli bond percolation, no phase transition occurs in $d=2$.

The relevance of the $n=2$ case suggests that a systematic understanding of $H^{2|2n}$ nonlinear sigma models for general $n$ is both natural and desirable, including their potential connections to other probabilistic models. 
In \cite{Crawford2021}, the model was analysed in dimension $d=2$ with the aim of understanding the decay of the two-point function, which, according to conventional physics wisdom, is expected to be exponential for all $\beta$, analogous to the behaviour observed in $O(N)$ models \cite{Polyakov1975,Caracciolo2017}.
Building on the works \cite{Spencer1977,Sabot2021}, the author established polynomial upper bounds on the two-point function; these bounds were derived using a certain positivity property of the partition function. Importantly, this same positivity property also strongly suggests the existence of a relation with a dual percolation model for arbitrary $n$.

Before delving into the details of our findings, a final remark is in order. One could likewise consider nonlinear sigma models with hyperbolic space $H^{2+2n'|2n+2n'}$, $n' \in \mathbb{N}$; however, also in these cases, one can reduce the problem to studying the marginals $H^{0|2n-2}$ via the supersymmetric localisation theorem, see Section \ref{sec:expansion}.

\subsection{Main result}
\label{sec:main result}

In this work, we focus on the high-temperature (corresponding to small $\beta$) phase of the $H^{2|2n}$ model and establish the following decay bounds on the  two-point function.

\begin{theorem} \label{main_theorem}
Let $\langle \,\cdot\,\rangle^{\Lambda}_{\beta,\veps,n}$ denote the Gibbs measure of the $H^{2|2n}$ nonlinear sigma model on
$\Lambda \subset \mathbb{Z}^{d}$ finite box at inverse temperature $\beta >0$ and external field $\veps >0$. For all $d \in \mathbb{N}$ there exists a constant $C_{0}>0$, depending only on the couplings $J$, such that for all $n>1$, $\beta,\veps > 0$ satisfying $\beta n C_{0} <1$
the following bounds on the two-point function hold true, uniformly in $\Lambda$ and $i_{0},j_{0} \in \Lambda$.
\begin{itemize}
    \item[1)] \underline{Nearest-neighbour interaction.} If  $J_{ij}=\ind{|i-j|=1},$ then
\begin{equation*}
 | \langle \xi_{i_{0}} \eta_{j_{0}} \rangle^{\Lambda}_{\beta,\veps,n} |\lesssim (C_{0}\beta n)^{|i_{0}-j_{0}|} \;,\qquad \quad \forall i_{0},j_{0} \in \Lambda \;.
\end{equation*}
 \item[2)] \underline{Exponentially decaying interaction.} If  $J_{ij} \lesssim \ee^{ - a \mathrm{dist}(i,j)}$ for some $a\! >\!0$ and for some metric $\mathrm{dist}$ such that $\sup_{i\in \mathbb{Z}^{d}}\!\sum_{j\in\mathbb{Z}^{d}} \ee^{-\frac{a}{2} \mathrm{dist}(i,j)}\! <\!\infty,$
then 
\begin{equation*}
 | \langle  \xi_{i_{0}} \eta_{j_{0}}\rangle^{\Lambda}_{\beta,\veps,n} |\lesssim (\beta n)^{\ind{i_{0} \neq j_{0}}} \ee^{-\frac{a}{2} \mathrm{dist}(i_{0},j_{0})} \;,\qquad \quad \forall i_{0},j_{0} \in \Lambda \;.
\end{equation*}
\item[3)] \underline{Polynomially decaying interaction.} If $J_{ij} \lesssim \frac{1}{(1+|i-j|)^{a}}$ for some $a > d$, then
\begin{equation*}
 | \langle  \xi_{i_{0}} \eta_{j_{0}}\rangle^{\Lambda}_{\beta,\veps,n} |\lesssim \frac{(\beta n)^{\ind{i_{0} \neq j_{0}}} }{(1+|i_{0}-j_{0}|)^{a}} \;,\qquad \quad \forall i_{0},j_{0} \in \Lambda \;.
\end{equation*}
\end{itemize}
\end{theorem}
\begin{remark}\label{rmk:main-result}\hspace{2cm}

\begin{itemize}
\item[i)] The dependence of the $\beta$-threshold on $n$ is optimal at the level of scaling. Indeed, for large $n$ each site carries order $n$ fermionic degrees of freedom, and the interaction terms involve a sum over these identical internal components via the symmetric product $\langle \cdot, \cdot\rangle$; as a consequence, the effective strength is proportional to $\beta n$.
This is entirely analogous to the large-$N$ normalisation in classical $O(N)$ models, where the coupling constant is rescaled as $\beta/N$ in order to obtain a non-trivial large-$N$ limit. If the Hamiltonian \eqref{eq:def-Hamilt} were normalised in this way, the small-$\beta$ regime would become uniform in $n$.

\item[ii)] Among the three cases, the nearest-neighbour interaction yields the strongest decay, which improves as $\beta \to 0$.
In case 2), natural choices of metric include $\mathrm{dist}(i,j)=|i-j|$ and $\mathrm{dist}(i,j)=\log(1+|i-j|)$. The former leads to genuine exponential decay, while the latter corresponds to polynomially decaying interactions.
In this second situation, however, the bound obtained from case 2) is suboptimal: as shown in case 3), one can in fact recover decay with the same exponent $a$ as the interaction.
Note that in case 2) we can extract exponential decay with rate arbitrarily close to $a,$ by choosing $\beta$ small enough.
 
\item[iii)] Note that the bound is uniform in $\veps$, in particular, unlike in the $n=1$ case, there is no divergence as $\veps \to 0$. As we will discuss, this is due to the surplus of Grassmann variables; see Section \ref{sec:expansion}, where we restrict to the study of the $H^{0|2n-2}$ model. However, the measure of the $H^{2|2n}$ model itself is not well-defined for $\veps =0$.
\end{itemize}
\end{remark}

\begin{remark}[Extensions to other observables]\label{rmk:main-extension}
By symmetry
\begin{equation*}
\langle \bar{\psi}_{i_0,\alpha} \psi_{j_0,\alpha} \rangle^{\Lambda}_{\beta,\veps,n} =\langle  \xi_{i_{0}} \eta_{j_{0}}\rangle^{\Lambda}_{\beta,\veps,n} \;, \qquad \forall \alpha=1,\dots,n-1 \;,
\end{equation*}
hence the bounds above hold also for these two-point functions. 

With minimal additional work, the proof extends to more general truncated correlation functions of the following type. Consider the model with nearest-neighbour couplings, let $D_k \subset \Lambda$ and $A_k$ be a polynomial in the Grassmann variables
$\{ \bar{\psi}_{j,\alpha},\psi_{j,\alpha} \}_{j \in D_k,\ \alpha =1,\dots,n-1}$, for $k=1,\dots,\ell$.
Then
\[
| \langle A_{1} ; \cdots ; A_{\ell}\rangle^{\Lambda}_{\beta,\veps,n} |
\;\lesssim\;
(C\beta n)^{\mathrm{dist}(D_{1},\dots,D_{\ell})},
\]
where $\mathrm{dist}(D_{1},\dots,D_{\ell})$ denotes the length of the minimal tree in $\mathbb{Z}^d$ connecting the domains $D_1,\dots,D_\ell$.

\end{remark}

Let us now compare more in detail our result with the existing literature.
Exponential decay of correlations in the high-temperature regime was previously established in \cite{DisertoriSpencer2010} for $n=1$ and  in \cite{Bauerschmidt2024} for $n=2$. 
 
In \cite{DisertoriSpencer2010}, the proof relies on the probabilistic representation in the $t$-field, see \eqref{eq:partition-function-VRJP}. Denoting by $\llangle \,\cdot\,\rrangle^{\Lambda}_{\beta,\veps,n}$ the probability measure in the $t$-variables, simple manipulations in horospherical coordinates \cite{DisertoriSpencerZirnbauer2010} yield
\begin{equation*}
 \langle  \eta_{j} \xi_{i}\rangle^{\Lambda}_{\beta,\veps,n} =  \llangle (D^{\Lambda}_{\veps})^{-1}_{ij}\rrangle^{\Lambda}_{\beta,\veps,n} \;,
\end{equation*}
where $D^{\Lambda}_{\veps}$ was introduced in \eqref{eq:def-D}. This is precisely the quantity studied in \cite{DisertoriSpencer2010} in the case of at least two-site pinning\footnote{Instead of a uniform external field $\veps$, they considered $\veps_{k} = \ind{k=i} + \ind{k=j}$, for the two fixed sites of interest $i,j \in \Lambda$.}. The strategy developed in \cite[Theorem 1]{DisertoriSpencer2010} to control the average of $(D^{\Lambda}_{\veps})^{-1}_{ij}$ crucially exploits the square root in the determinant appearing in the measure.  While an extension of this strategy to general $n$ cannot be ruled out, such an extension is not straightforward and does not follow directly from the arguments in \cite{DisertoriSpencer2010}. The same structural feature is crucial in their analysis of other observables \cite[Theorem 2]{DisertoriSpencer2010}, involving single-site pinning\footnote{That is, $\veps_{k} = \ind{k=i}$, for some fixed site $i \in \Lambda$.}.
We also note that the decay rate obtained in \cite{DisertoriSpencer2010} is weaker, namely of the form
$(\ln \beta^{-1} \sqrt{\beta})^{|i-j|}$. On the other hand, the purely probabilistic formulation of the model allowed the authors to cover any value of $\beta$ in $d=1$, whereas our result holds for small $\beta$ only.

Concerning the $n=2$ case, the proof in \cite{Bauerschmidt2024} is based on the arboreal-gas representation of the model. In this formulation, the two-point function is related to the connection probability in the arboreal gas via the identities\footnote{Compared to \cite{Bauerschmidt2024}, the order of integration for $\eta$ and $\xi$ is reversed, hence the different expression for the two-point function.}
\begin{equation*}
\mathbb{P}^{\Lambda}_{\beta,\veps}[0 \leftrightarrow j, 0 \not \leftrightarrow \mathfrak{g}] = \langle \eta_{j}\xi_{0}\rangle^{\Lambda}_{\beta,\veps,2} \;,\qquad \qquad
\mathbb{P}^{\Lambda}_{\beta,0}[0 \leftrightarrow j] = \lim_{\veps \to 0}  \langle\eta_{j} \xi_{0}\rangle^{\Lambda}_{\beta,\veps,2} \;,
\end{equation*}
where $\mathfrak{g}$ is the graveyard point, see \cite{Bauerschmidt2021,Bauerschmidt2024}. This identification, allowed the authors to combine stochastic domination together with  general results on the subcritical bond percolation to establish exponential decay. While this approach is very elegant and powerful, it does not provide quantitative estimates on the decay rate. Moreover, its extension to arbitrary $n$ seems out of reach, as no analogous percolation representation is available \cite{Crawford2021}.

In contrast to the mentioned works for $n=1,2$, our proof is based on the study of the marginal purely fermionic $H^{0|2n-2}$ model by a high-temperature cluster expansion. Although the high-temperature regime is often regarded as straightforward, our result is not just an exercise in cluster expansions. 
First of all, because the expansion is not performed around a Gaussian measure, standard tools from Gaussian Grassmann integration are not directly applicable. 
This difficulty can be overcome through a careful use of Grassmann norms; however, this alone is insufficient to obtain the optimal dependence on $n$ as $n \to \infty$.
A central contribution of our work is therefore the derivation of sharp $n$-dependent bounds, achieved through a refined combination of Grassmann-norm estimates and detailed combinatorial analysis.

\paragraph*{Structure of the paper} In Section~\ref{sec:proof-main-thm} we derive the high-temperature expansion into a polymer system (see Lemma \ref{lemma:high-temperature} and Proposition \ref{expansion_lemma}), state the optimal activity bounds (see Proposition \ref{tree_estimate}), and use them to prove the main theorem.
Section~\ref{sec:norm_estimate} contains the core technical analysis, namely the proof of the optimal activity estimates; these rely crucially on a sharp control of the single-site partition function, see in particular Section~\ref{sec:single-point}.
For completeness, Appendix~\ref{appendix_expansion} provides the details of the high-temperature expansion, while Appendix~\ref{sec:arboreal gas} gives a new proof of the representation of the partition function, in the case $n=2,$ as the partition function of the arboreal gas.

\section{Proof of the main theorem}
\label{sec:proof-main-thm}

\subsection{High-temperature expansion}
\label{sec:expansion}

For convenience, we shall henceforth fix
\begin{equation*}
m:=n-1 \in \mathbb{N}_{\geq 1} \;.
\end{equation*}
As a first step, we reduce the problem to studying the purely Grassmann $H^{0|2m}$ nonlinear sigma model. To this end, we now consider supervectors of the form
\begin{equation*}
v_{j}=(z_{j}, \psi_{j}) \;, \qquad \psi_{j}:= (\bar{\psi}_{j,1}, \psi_{j,1}, \dots ,\bar{\psi}_{j,m}, \psi_{j,m}) \;,
\end{equation*}
compare with \eqref{eq:def-u}, \eqref{eq:z} and \eqref{eq:symmetric-product}, with
\begin{equation}\label{eq:exp-z-psi}
z_{j} = \sqrt{1 + \psi_{j} \cdot \psi_{j}} = \sum_{k=0}^{n} \tbinom{1/2}{k}(\psi_{j} \cdot \psi_{j})^{k} \;,
\end{equation}
where $\tbinom{1/2}{k} = \frac{1}{k!}\frac{1}{2}\big(\frac{1}{2} - 1 \big)\cdots \big(\frac{1}{2} - k+1\big)$ is the generalised binomial coefficient. The symmetric product on $H^{0|2m}$, by slight abuse of notation, is $
\langle v , v' \rangle := -z z' + \psi \cdot \psi'$. We then define the Gibbs measure of the $H^{0|2m}$ model on a finite box $\Lambda \subset \mathbb{Z}^{d}$, at inverse temperature $\beta>0$ and external field $\veps \geq  0$  by setting
\begin{equation}\label{eq:Gibbs-measure-Grassmann}
\big[ F\big]^{\Lambda}_{\beta,\veps,m} := \frac{1}{\tilde{Z}_{\beta,\veps,m}^{\Lambda}} \int
 \dd \tilde\mu_{\Lambda} \ee^{-\tilde{H}^{\Lambda}_{\beta,\veps}} F \;, \qquad \tilde{Z}^{\Lambda}_{\beta,\veps,m} : = \int
 \dd \tilde\mu_{\Lambda} \ee^{-\tilde{H}^{\Lambda}_{\beta,\veps}}
\end{equation}
for any observable $F=F(\psi)$. Here, we let
\begin{equation*}
\int \dd \tilde{\mu}_{\Lambda} \; \cdot \; := \prod_{j \in \Lambda} \Big[ \prod_{\alpha =1}^{m}\partial_{\bar{\psi}_{j,\alpha}}\partial_{\psi_{j,\alpha}}\Big]  \frac{1}{z_{j}} \; \cdot \;
\end{equation*}
denote the Grassmann integration form, while the Hamiltonian
$\tilde{H}^{\Lambda}_{\beta,\veps}$ is as in \eqref{eq:def-Hamilt}, with the substitution $u_{j} \to v_{j}$
\begin{equation}\label{eq:tilde-H}
\begin{split}
\tilde{H}^{\Lambda}_{\beta,\veps}& = \frac{\beta}{2} \sum_{\{i,j\} \in E_\Lambda} J_{ij} \langle v_{i} -v_{j}, v_{i} -v_{j} \rangle + \veps \sum_{j \in \Lambda} (z_{j}-1) \;.
\end{split}
\end{equation}
The relation between the $H^{2|2n}$ and the $H^{0|2m}$ models is elucidated in the following lemma.
\begin{lemma}\label{lemma:reduction}
For any $\Lambda \subset \mathbb{Z}^{d}$, $\beta,\veps >0$, $m \geq 1$, $\alpha =1,\dots, m$ and $i_{0},j_{0} \in \Lambda$ we have
\begin{equation*}
Z^{\Lambda}_{\beta,\veps,m+1}=
\tilde{Z}^{\Lambda}_{\beta,\veps,m}  \;, \qquad \qquad
 \langle  \xi_{i_{0}} \eta_{j_{0}}\rangle^{\Lambda}_{\beta,\veps,m+1} = \big[ \bar{\psi}_{i_{0},\alpha} \psi_{j_{0},\alpha}\big]^{\Lambda}_{\beta,\veps,m} \;.
\end{equation*}
\end{lemma}
\begin{proof}
The integrand in the definition of $Z^{\Lambda}_{\beta,\veps,m+1}$ is supersymmetric with respect to the variables $x,y,\xi,\eta$ and has sufficient integrability for $\veps >0$ \cite{DisertoriSpencerZirnbauer2010}. Then, the claim follows by application of the supersymmetric localisation theorem \cite{Schwarz1997}, see also \cite{Crawford2021}. 
The same argument applies to the two-point function, since
by symmetry under exchange of the Grassmann variables, we have $\langle  \xi_{i_{0}} \eta_{j_{0}}\rangle^{\Lambda}_{\beta,\veps,m+1} = \langle \bar{\psi}_{i_{0},\alpha} \psi_{j_{0},\alpha}\rangle^{\Lambda}_{\beta,\veps,m+1}$.
\end{proof}

Let us now have a closer look at the Grassmann representation of the $H^{0|2m}$ model; we can write the Hamiltonian \eqref{eq:tilde-H} as 
\begin{equation}
\begin{split}
\tilde{H}^{\Lambda}_{\beta,\veps} = -\beta \sum_{\{i,j\} \in E_\Lambda} J_{ij} \big(\psi_{i} \cdot \psi_{j}  -(z_{i}z_{j}-1)\big) + \veps \sum_{j \in \Lambda} (z_{j}-1) \;,
\end{split}
\end{equation}
where recall that there is no contribution from $i=j$ by construction.
For any type of interaction this Hamiltonian defines a gradient-type model, which is a polynomial perturbation of a Gaussian.
In the regime of large $\beta$, the Gaussian part is expected to dominate; this has been proved for $m=1$, see \cite{Bauerschmidt2024}. In this case, the perturbation is a quartic polynomial of gradient type, see \eqref{eq:arboreal-Z2}, and this fact is crucial for the proof. By contrast, in the regime of small $\beta$, the factorised part dominates, and this gradient structure becomes less important.
Actually, expanding the $z$'s as polynomials in $\psi \cdot \psi$ is not particularly useful at this stage either. Instead, it is more effective to keep track of which contributions preserve a factorised structure and which terms (weakly) break factorisation. This perspective underlies the general philosophy of high-temperature expansions and will be central in our analysis.

We adjust our notation accordingly. We set
\begin{equation}\label{eq:def-W-nu}
\cw_{ij}=\beta  J_{ij}(-1-\psi_i \cdot \psi_j+z_iz_j) \;, \quad \forall \{i,j\} \in E_{\Lambda}, 
\qquad
\cv_j=\frac{1}{z_j}\ee^{-\veps (z_j-1)}\;, \quad \forall j \in \Lambda \;.
\end{equation}
and introduce the notation, for any subset $Y \subset \Lambda$
\begin{equation}\label{eq:fact-notation}
\cw(Y) := \sum_{ \{i,j\} \in E_Y} \cw_{ij} \;, 
\qquad 
\cv^{Y} := \prod_{j \in Y} \cv_{j}  \;,
\qquad 
\int \dd \psi_{Y} \; \cdot \;:=  \prod_{j \in Y} \Big[ \prod_{\alpha =1}^{m}\partial_{\bar{\psi}_{j,\alpha}}\partial_{\psi_{j,\alpha}}\Big] \; \cdot \;\;,
\end{equation}
where, for brevity, we keep the dependence on $\beta$ and $\veps$ implicit. We then introduce the following factorised Grassmann integration form, for any $Y \subset \Lambda$
\begin{equation}\label{eq:reference-integration-form}
\int \dd \nu_{Y} \; \cdot \; :=  \int\dd \psi_{Y} \,\cv^{Y}  \; \cdot \; \;.
\end{equation}
With this notation, we can write, compare with \eqref{eq:Gibbs-measure-Grassmann},
\begin{equation*}
\tilde{Z}^{\Lambda}_{\beta,\veps,m} = \int \dd \nu_{\Lambda} \ee^{-\cw(\Lambda)} \;.
\end{equation*}
The measure in \eqref{eq:reference-integration-form} is the reference factorised measure for the high-temperature cluster expansion, while the Gibbs factor $\ee^{-\cw(\Lambda)}$ weakly breaks the factorisation for $\beta$ small.

To study correlation functions, we introduce the generating function for our model, which is tantamount to the Laplace transform of the reference measure. The external variables are Grassmann variables which we denote by $\{ \bar{\rho}_{j,\alpha},\rho_{j,\alpha} \}_{j \in \Lambda, \alpha =1,\dots,m}$. For any $Y \subset \Lambda$, we use the shorthand
\begin{equation*}
(\psi\cdot \rho)_{Y}:=\sum_{j \in Y} \psi_{j} \cdot \rho_{j} \qquad \rho_{j} := ( \bar{\rho}_{j,1},\rho_{j,1}, \dots,\bar{\rho}_{j,m},\rho_{j,m} ) \;,
\end{equation*}
where $\psi_{j} \cdot \rho_{j}$ is the symmetric product in \eqref{eq:symmetric-product}. The generating function is defined by 
\begin{equation}\label{eq:generating-function}
\tilde{Z}^{\Lambda}_{\beta,\veps,m}(\rho) := \int \dd \nu_{\Lambda} \, \ee^{-\cw(\Lambda)} \ee^{(\psi \cdot \rho)_{\Lambda}} \;.
\end{equation}
Note that $\tilde{Z}^{\Lambda}_{\beta,\veps,m}(0) = \tilde{Z}^{\Lambda}_{\beta,\veps,m}$. 
Connected correlation functions are obtained by derivation with respect to the external variables; in particular, 
the two-point function can be written as
\begin{equation}\label{eq:def2point-function}
\big[ \bar{\psi}_{i,\alpha} \psi_{j,\alpha}\big]^{\Lambda}_{\beta,\veps,m}  = \big[ \bar{\psi}_{i,\alpha} \psi_{j,\alpha}\big]^{\Lambda}_{\beta,\veps,m} - \big[ \bar{\psi}_{i,\alpha}\big]^{\Lambda}_{\beta,\veps,m}\big[\psi_{j,\alpha}\big]^{\Lambda}_{\beta,\veps,m} =
\frac{\partial}{\partial \rho_{i,\alpha}} \frac{\partial}{\partial \bar{\rho}_{j,\alpha}} \ln \tilde{Z}^{\Lambda}_{\beta,\veps,m}(\rho) \Big|_{\rho = 0} \;,
\end{equation}
since $\big[ \bar{\psi}_{i,\alpha}\big]^{\Lambda}_{\beta,\veps,m} = \big[\psi_{j,\alpha}\big]^{\Lambda}_{\beta,\veps,m} = 0$ by symmetry. Here, for any function $F(\rho)$ setting $\rho = 0$ corresponds to taking zero-order coefficient of the Taylor expansion in $\rho$. For later convenience, we introduce, respectively, the single-site generating function and normalisation
\begin{equation}\label{eq:single-site-normalisation}
\tilde{Z}_{\veps,m}^{\{ j\}}(\rho) := \int \dd \nu_{\{ j\}} \ee^{\psi_{j} \cdot \rho_{j}} \;,
\qquad \qquad
\tilde{Z}_{\veps,m} := \int \dd \nu_{\{ j\}} \;,
\end{equation}
the latter being independent of $j \in \Lambda$.

A high-temperature representation for the generating function is obtained by expanding  the Gibbs weight $\exp(-\cw(\Lambda))$ into its connected components. These are defined by setting
\begin{equation}\label{eq:expansion0}
(\ee^{-\cw(\Lambda)})_{\mathrm{conn}} := \ee^{-\cw(\Lambda)} =  1 \;, \qquad \text{for} \;\; |\Lambda|=1 \;,
\end{equation}
and for $|\Lambda| \geq 2$ by recursively requiring
\begin{equation}\label{eq:expansion}
 \ee^{-\cw(\Lambda)}=:\sum_{\Pi \;\mathrm{part}\; \Lambda} \prod_{Y \in \Pi} (\ee^{-\cw(Y)})_{\mathrm{conn}} \;,
\end{equation} 
where the sum runs over partitions $\Pi$ of $\Lambda$.
As a consequence of the expansion in \eqref{eq:expansion}, the
generating function admits a representation as a hard-core polymer gas.
\begin{lemma}[High-temperature expansion]\label{lemma:high-temperature}
With the notation set above, for any $\Lambda \subset \mathbb{Z}^{d}$, $\beta >0$, $\veps\geq 0$, as long as $\tilde{Z}_{\veps,m} \neq 0$ it holds that
\begin{equation}\label{eq:high-temp}
 \frac{\tilde Z^{\Lambda}_{\beta,\veps,m}(\rho)}{\prod_{j \in \Lambda} \tilde{Z}_{\veps,m}^{\{ j\}}(\rho)}=1 +\sum_{N \geq 1}\frac{1}{N!}\sum_{\substack{Y_1,\ldots,Y_N\subset \Lambda\\ |Y_l|>1}} \Bigg[\prod_{l=1}^N \mathcal{K}(Y_l) \Bigg] \varphi(Y_1,\ldots,Y_N) \;,
\end{equation}
where $\varphi(Y_1,\ldots,Y_N)=
\prod_{ l\neq m=1}^{N}\ind{Y_l\cap Y_m=\varnothing}$  is the hard-core interaction between sets, and where
\begin{equation}\label{eq:defKgener}
\mathcal{K}(Y):=\frac{\int \dd \nu_{Y}\ (\ee^{-\cw(Y)})_{\mathrm{conn}}\, \ee^{(\psi \cdot\rho)_Y}}{\prod_{j \in Y} \tilde{Z}_{\veps,m}^{\{ j\}}(\rho)} \;,
\end{equation}
takes value in the Grassmann algebra generated by $\{ \bar{\rho}_{j,\alpha},\rho_{j,\alpha} \}_{j \in \Lambda, \alpha = 1,\dots,m}$.
\end{lemma} 
\begin{remark}\hspace{2cm}

\begin{itemize}
\item[i)]  
Above we can also take $\veps =0$, because we are now working with Grassmann variables only,
so that $\tilde Z^{\Lambda}_{\beta,\veps,m}(\rho)$ and $\tilde Z^{\{j\}}_{\veps,m}(\rho)$ are always well-defined.
Note that $\tilde Z^{\{j\}}_{\veps,m}(\rho)$  is invertible provided that $\tilde{Z}_{\veps,m} \neq 0.$ In Section \ref{sec:single-point}
below we will show that $\tilde{Z}_{\veps,m} > 0$ for any $\veps \geq 0$. 
\item[ii)] Although the sum in \eqref{eq:high-temp} is written as $\sum_{N\geq 1}$, it is in fact finite. Indeed, as will be clear from the proof, one necessarily 
has $N \leq \lfloor |\Lambda|/2\rfloor$, 
since this is the maximal number of disjoint subsets of $\Lambda$ of cardinality at least two.
\end{itemize}
\end{remark}
\begin{proof}
Using \eqref{eq:expansion}, we can write
    \[
     \frac{\tilde Z^{\Lambda}_{\beta,\veps,m}(\rho)}{\prod_{j \in \Lambda} \tilde{Z}_{\veps,m}^{\{ j\}}(\rho)} = \sum_{\Pi \;\mathrm{part}\; \Lambda} \prod_{Y \in \Pi} \mathcal{K}(Y) \;.
    \]
For $Y=\{j \}$ we have, using \eqref{eq:expansion0}, $(\ee^{-\cw(Y)})_{\mathrm{conn}}=1$ and hence
\[
\mathcal{K}(Y) = \frac{\int \dd \nu_{{j}}\, \ee^{ (\psi_{j} \cdot \rho_{j})}}{ \tilde{Z}_{\veps,m}^{\{ j\}}(\rho)} =1.
\]
Thus  singletons can be omitted and  the sum over partitions can be rewritten as a sum of pairwise disjoint sets
of size at least two:
\begin{equation*}
\sum_{\Pi \;\mathrm{part}\; \Lambda} \prod_{Y \in \Pi} \mathcal{K}(Y) = 1 + \sum_{N \geq 1} \frac{1}{N!}\hspace{-0,4cm}\sum_{\substack{Y_1,\ldots,Y_N\subset \Lambda\\ \text{pairwise disjoint and }|Y_l|>1}}\prod_{Y \in \Pi} \mathcal{K}(Y)
\end{equation*}
where the $1/N!$ is there to compensate overcounting due to the permutation of the $Y$'s. The claim then follows by imposing the pairwise disjointness via the hard-core function.
\end{proof} 
In the language of statistical mechanics, the subsets $Y \subset \Lambda$ are referred to as \emph{polymers}, and $\mathcal{K}(Y)$ as their
\emph{activity}. 
Starting from this representation, bounds on correlation functions are obtained by taking logarithms and estimating derivatives of the resulting polymer expansion, following standard arguments \cite{Brydges:1984vu}. As a result, we obtain the following estimate.
\begin{proposition}\label{expansion_lemma}
For any $\Lambda \subset \mathbb{Z}^{d}$ finite, $\beta >0$, $\veps \geq 0$, $m \in \mathbb{N}$, $i_{0},j_{0} \in \Lambda$ and $\alpha = 1,\dots, m$ the following bound holds true for some constant $C>1$
\begin{equation}\label{eq:expansion_lemma}
\big|\big[ \bar{\psi}_{i_{0},\alpha} \psi_{j_{0},\alpha}\big]^{\Lambda}_{\beta,\veps,m} \big| \leq  
\Bigg( \sum_{\substack{Y\subset\Lambda:\\Y \ni i_{0},j_{0}}}|K_{i_{0},j_{0}}(Y)| C^{|Y|}\Bigg)
\sum_{N\geq 1}
\Bigg(\sup_{k\in \Lambda}\sum_{\substack{Y\subset \Lambda: \\ Y \ni k, \; |Y|>1}}|K(Y)|C^{|Y|}\Bigg)^{N-1}\; ,
\end{equation}
where the polymer activities $K (Y),K_{i_{0},j_{0}}(Y)$,  $\forall Y \subset \Lambda$ and $i_{0},j_{0} \in \Lambda,$ are defined by
\begin{align}\label{eq:pol-external-field}
K(Y) & := \mathcal{K}(Y)\big|_{\rho  = 0} =\frac{\int \dd \nu_{Y}\ (\ee^{-\cw(Y)})_{\mathrm{conn}}}{ \tilde{Z}^{|Y|}_{\veps,m}}  \\
K_{i_{0},j_{0}}(Y)&:= \frac{\partial}{\partial \rho_{i_{0},\alpha}}\frac{\partial}{\partial \bar{\rho}_{j_{0},\alpha}}
\Big[\ind{j_{0}=i_{0} } \log \tilde{Z}_{\veps,m}^{\{ i_{0}\}}(\rho) + \mathcal{K}(Y)\Big]_{\rho=0} \hspace{-0,2cm}=
\frac{\int \dd \nu_{Y}\ (\ee^{-\cw(Y)})_{\mathrm{conn}} \bar{\psi}_{i_{0},\alpha} \psi_{j_{0},\alpha}}{ \tilde{Z}^{|Y|}_{\veps,m}}\; .\nonumber
\end{align}
\end{proposition}
The proof is standard but somewhat involved; for the reader's convenience, we provide the details in Appendix~\ref{appendix_expansion}.

\subsection{Grassmann norms}
\label{sec:activity-estimates}

In order to control the sums in \eqref{eq:expansion_lemma}, we need suitable estimates for the activities
$|K(Y)|$ and $|K_{i_{0},j_{0}}(Y)|$. This is where the specific structure of the model, and in particular
its Grassmann nature, plays a crucial role.

Although Grassmann integrals can in principle be always evaluated exactly, their complexity grows rapidly with the number of variables. Moreover, in the present setting standard tools from Gaussian Grassmann integration are not available, since the underlying reference measure is factorised rather than Gaussian.
Nevertheless, useful bounds can still be obtained by estimating Berezin integrals using the $\ell^{1}$-type norm
\cite{Fresta2021}, which we now recall. 

For brevity, we denote by $\mathcal{G}^{\Lambda}_{m}$ the  Grassmann algebra (real or complex) with generators $\{\bar{\psi}_{j,\alpha}, \psi_{j,\alpha} \}_{j \in \Lambda, \alpha =1,\dots,m}$.
This algebra is a finite linear space with canonical basis elements
\begin{equation*}
\psi^{I}\bar{\psi}^{\bar{I}}:= \Big[\prod_{(i,\alpha) \in I}  \psi_{i,\alpha} \Big] \Big[\prod_{(j,\alpha) \in \bar{I}} \bar\psi_{j,\alpha} \Big] \;,\qquad \quad I,\bar I \subset \Lambda \times \{1,\dots,m \} \;.
\end{equation*}
Accordingly, any element $f \in \mathcal{G}^{\Lambda}_{m}$ can be uniquely written as
\begin{equation*}
f =  \sum_{I,\bar{I}} f_{I,\bar{I}}\psi^{I}\bar{\psi}^{\bar{I}}
\end{equation*}
and identified as the sequence $\big(f_{I,\bar{I}} \big)_{I,\bar I}$. Equipped with the $\ell^{1}$-type norm on such sequences, the Grassmann algebra $\mathcal{G}^{\Lambda}_{m}$ becomes a normed algebra. The following result  is proved in  \cite[Lemma 2.4]{Fresta2021}.
\begin{lemma}[$\ell^{1}$-type norm {\cite[Lemma 2.4]{Fresta2021}}]\label{lemma:norm}
For any $f \in \mathcal{G}^{\Lambda}_{m}$, $f = \sum_{I,\bar{I}} f_{I,\bar{I}}\psi^{I}\bar{\psi}^{\bar{I}}$, set
\begin{equation*}
\| f\|:= \sum_{I,\bar{I}} |f_{I,\bar{I}}| \;.
\end{equation*}
Then, $\mathcal{G}^{\Lambda}_{m}$ equipped with $\| \cdot \|$ is a normed algebra, that is,
\begin{equation*}
\| 1\| = 1 \qquad \qquad \| f g\| \leq \| f\| \, \|g\|\;, \qquad \forall f,g \in \mathcal{G}^{\Lambda}_{m} \;.
\end{equation*}
Furthermore, the following bound holds for the Grassmann integral (see also \eqref{eq:fact-notation})
\begin{equation*}
\Bigg|\int \dd \psi_{\Lambda} \,f \Bigg| \leq \| f \| \;, \qquad \forall f \in \mathcal{G}^{\Lambda}_{m} \;.
\end{equation*}
\end{lemma}

Although these Grassmann norms yield effective bounds on $|K(Y)|$ and $|K_{i_{0},j_{0}}(Y)|$, they are not sufficient to establish Theorem~\ref{main_theorem} with optimal dependence on $m$, for two reasons. First, the normalisation $\tilde{Z}_{\veps,m}$ must be controlled by more refined means, since an upper bound alone is inadequate for our purposes. Once more, this is non-trivial because we cannot use Gaussian-integration tools at our disposal and the complexity of the integral grows rapidly with $m$.
Second, and more subtly, a naive estimation of the integral fails to produce the optimal dependence on $m$.
Let us briefly see why this is the case and how we overcome this difficulty. 

Let us consider $K(Y)$ with nearest-neighbour interaction for brevity. By Lemma \ref{lemma:norm} we may estimate
\[
   \left|\int \dd \psi_{Y}\ (\ee^{-\cw(Y)})_{\mathrm{conn}}\cv^{Y}\right | \leq \left \|(\ee^{-\cw(Y)})_{\mathrm{conn}} \right\|  \prod_{j \in Y}\big\| \cv_{j} \big \| \;.
\]
Expanding $(\ee^{-\cw(Y)})_{\mathrm{conn}}$ by means of Theorem \ref{thm:BBF-formula}, the minimal tree connecting $i_{0}$ and $j_{0}$ through nearest-neighbour edges has at least $|i_{0} - j_{0}|$ edges, leading to the sought factor $\beta^{|i_{0} - j_{0}|}$, and thus the exponential decay in the distance $|i_{0} - j_{0}|$ for $\beta $ sufficiently small.

However, we are still left with controlling the norm of the polynomials in the $z$'s coming both from the interaction term $\cw_{ij}$ and from the single site measure $\nu_{j}$, see \eqref{eq:def-W-nu}. Obtaining bounds with an optimal dependence in $m$ is tricky because
\begin{equation}\label{eq:even_z_bound}
 \| z_{j}\| \sim m^{m} \;, 
\end{equation}
whereas even powers have a better estimate
\begin{equation*}
 \|z_{j}^{2}\|=\left \|1+\psi_{j} \cdot \psi_{j} \right \|=1+2m \;.
\end{equation*}
Thus, estimating the norm of a monomial as $\|z_{j}^{l}\| \leq \|z_{j}\|^{l} \sim m^{m l}$ leads to suboptimal growth in $m$.
To circumvent this issue, we must carefully keep track of the powers in $z$'s in $(\ee^{-\cw(Y)})_{\mathrm{conn}}$ and, where appropriate, take into account the factor $1/z_{j}$ coming from $\nu_{j}$, see Section \ref{sec:optimal-activity} for details.
This refinement still leads us to the problem of estimating the norm of either $\nu_{j}$ or $z_{j} \nu_{j}$ in a way that matches, in terms of dependence in $m$, the exact computation of the normalisation $\tilde{Z}_{\veps,m}$. Achieving this requires further delicate combinatorial analysis, which is carried out in Section \ref{sec:single-point}.
\medskip

\subsection{Proof of the main result from the activity estimates}

In Section \ref{sec:norm_estimate} we prove the following activity bounds with optimal dependence on $m$:
\begin{proposition}[Activity bounds]\label{tree_estimate}
Let $\beta m \leq 1$. Then, there exists a constant $C>1$ independent of $\beta$ and $m$ such that for all $\veps \geq 0$
\begin{equation*}
\begin{split}
      |K(Y)| &\leq C^{|Y|}(\beta m)^{|Y|-1}\sum_{T \; \mathrm{on} \; Y}\prod_{ \{l,l' \}\in E(T)} J_{ll'}  \;,
    \\
 |K_{i_{0},j_{0}}(Y)|   &\leq C^{|Y|}(\beta m)^{|Y|-1}\sum_{T \; \mathrm{on} \; Y}\prod_{ \{l,l' \}\in E(T)} J_{ll'}  \;,
\end{split}
\end{equation*}
where the sum runs over the spanning trees $T$ on $Y$, and where $E(T)$ denotes the set of edges of $T$.
\end{proposition} 
Note that $\beta m\leq 1$ is chosen for convenience and could be replaced by  $\beta m$ bounded by a fixed constant.
Together with Proposition \ref{expansion_lemma}, the activity bounds allow us to prove Theorem \ref{main_theorem} by standard tree estimates, which we here report for the reader's convenience.
\begin{proof}[Proof of Theorem \ref{main_theorem}]
In the following, we use $C$ to denote generic constants that are independent of $\beta$, $\veps$ and $m$; their value may change  from line to line.

 By Lemma \ref{lemma:reduction} and Proposition \ref{expansion_lemma}, we have $\forall \varepsilon >0,$
\begin{equation}\label{eq:bound2pointxieta}
|\langle  \xi_{i_{0}} \eta_{j_{0}}\rangle^{\Lambda}_{\beta,\veps,m+1}|=
|\big[ \bar{\psi}_{i_{0},\alpha} \psi_{j_{0},\alpha}\big]^{\Lambda}_{\beta,\veps,m}|\leq
\mathcal{A}_{i_{0},j_{0}} \sum_{N\geq 1} \mathcal{B}^{N-1},
\end{equation}
where 
\begin{equation*}
\mathcal{A}_{i_{0},j_{0}}:=\sum_{\substack{Y\subset\Lambda:\\Y \ni j_{0},i_{0}}}|K_{i_{0},j_{0}}(Y)|C^{|Y|} \;,
\qquad \qquad \mathcal{B}:=\sup_{k\in \Lambda}\sum_{\substack{Y\subset \Lambda:\\ Y \ni k, \; |Y| >1}}|K(Y)|C^{|Y|} \;.
\end{equation*}
Note that while there is no constraint on the size of $Y$ for $\mathcal{A}_{i_{0},j_{0}}$, we must have $|Y|>1$ for $\mathcal{B}$, allowing us to gain an extra small factor. The expressions $\mathcal{A}_{i_{0},j_{0}},\mathcal{B}$ are well-defined also for $\varepsilon =0.$
Below, we prove that there exists some constant $C_{0}>1$ such that, for all $\beta>0$, $\veps\geq 0$, $m \in \mathbb{N}$ such that
$C_{0}\beta n=C_{0}\beta (m+1)  < 1,$
it holds
\begin{equation}\label{decay_C3}
     \mathcal{A}_{i_{0},j_{0}} 
       \lesssim
        \begin{cases}
                (C_{0}\beta m)^{|i_0-j_0|} \quad &\text{if }\quad J_{ij}=\ind{|i-j|=1} \;,
                \\[4pt]
             (\beta m)^{\ind{i_0\neq j_0}}    \ee^{-\frac{a}{2} \mathrm{dist}(i_{0},j_{0})}  &\text{if }\quad J_{ij} \lesssim \ee^{-a \mathrm{dist}(i_{0},j_{0})},
               \\[4pt]
              (\beta m)^{\ind{i_0\neq j_0}}   \frac{1}{(1+|i_0-j_0|)^{a}}    &\text{if }\quad J_{ij} \lesssim \frac{1}{(1+|i-j|)^{a}}.
        \end{cases}    
\end{equation}
A simple adaptation of the argument, crucially using that $|Y|\!>\!1,$ gives  $\mathcal{B}\leq C_{0}\beta m<1.$
Inserting these estimates in \eqref{eq:bound2pointxieta} yields the  claims in Theorem \ref{main_theorem} with $n=m+1$.
\vspace{0,2cm}

We prove now \eqref{decay_C3}. Let $l_{0}:=|\{i_{0},j_{0} \}|\in \{1,2\}$ and let $\tilde\Lambda:= \Lambda \setminus \{i_{0},j_{0}\}$. For convenience, we also denote the elements of $\{i_{0},j_{0} \}$ as $x_{1},\dots,x_{l_{0}}$. If $Y \supset \{i_{0},j_{0}\}$ strictly, we label its additional points as
$x_{l_0+1},\dots,x_{N}$, $x_{l_{0}+1},\dots,x_{N} \in \Lambda \setminus \{i_{0},j_{0} \}$ distinct, where we abridge $N:=|Y|$.
Inserting the activity bounds from Proposition \ref{tree_estimate} in $\mathcal{A}_{i_{0},j_{0}}$ yields
\begin{align}
\mathcal{A}_{i_{0},j_{0}}
&\leq \sum_{\substack{Y\subset \Lambda \\ Y \ni j_0,i_0} }\sum_{T \text{ on }Y }
C^{|Y|}(\beta m)^{|Y|-1}\prod_{ \{x,x' \}\in E(T)} J_{xx'} 
\nonumber\\
        &\leq  C\sum_{N \geq l_0}\frac{(C\beta m)^{N-1}}{(N-l_0)!}\sum_{\substack{ x_{l_{0}+1},\dots,x_{N} \in \tilde\Lambda \\ \text{distinct}}}\sum_{\substack{T  \text{ on } \\ \{x_{1},\dots,x_{N}\}}} \prod_{ \{x,x' \}\in E(T)}J_{xx'}
       \nonumber \\
        &\leq  C \sum_{N \geq l_{0}}\frac{(C\beta m)^{N-1}}{|N-2|!}N^{N-2} \sup_{\substack{T  \text{ on } \\  \{1,\dotsc ,N \}}}\sum_{\substack{ x_{l_{0}+1},\dots,x_{N} \in \tilde\Lambda \\ \text{distinct}}}\prod_{ \{l,l' \}\in E(T)} J_{x_{l} x_{l'}} \;,\label{eq:Acalbound}
\end{align}
where, by Cayley's theorem, the number of labelled trees is bounded by $N^{N-2}$. In the last step, we switch summation over the trees and the points $\{x_{l_{0}+1},\dots,x_{N} \}$, in order to take the supremum over the trees; this requires us to adjust the notation, so that the trees are no longer on $\{x_{l_{0}+1},\dots,x_{N} \}$ but on $\{1,\dots,N \}$.

We now discuss the different choices of interaction separately.
\medskip

\underline{\textit{Nearest-neighbour interaction:}} set $J_{ij}=\ind{|i-j|=1}.$

The product $\prod_{ \{l,l' \}\in E(T)}J_{x_{l} x_{l'}}$ is non-zero only if the tree $T$ contains a nearest-neighbour path connecting $i_0$ and $j_0$. Consequently, such a tree must have at least $|i_0-j_0|+1$ vertices,
hence $N\geq |i_0-j_0|+1 \geq l_{0}$.

Moreover, since each edge can be chosen in at most $2d$ ways, we have
\[
       \sum_{\substack{ x_{l_{0}+1},\dots,x_{N} \in \tilde\Lambda  \\ \text{distinct}}}\prod_{ \{l,l' \}\in E(T)} J_{x_{l} x_{l'}}\leq (2d)^N \;.
\]
Therefore,
\begin{equation*}
\mathcal{A}_{j_{0},i_{0}}
   \le  C \sum_{N \geq |i_0-j_0|+1 }
        \frac{(C\beta m)^{N-1}}{|N-2|!}\, N^{N-2} (2d)^{N}
\lesssim (C_{0} \beta m)^{ |i_{0} - j_{0}|} \;.
\end{equation*}

\underline{\textit{Exponentially-decaying interaction:}} set  $J_{ij} \lesssim \ee^{-a \mathrm{dist}(i_{0},j_{0})},$ with $a>0.$

If $\sup_{i\in \mathbb{Z}^{d}}\!\sum_{j\in\mathbb{Z}^{d}} \ee^{-\frac{a}{2} \mathrm{dist}(i,j)}\! <\!\infty,$ 
we can extract a factor $\ee^{-\frac{a}{2}L(T)}$, where $L(T)$ is the length of the tree $T$ in terms of the metric $\mathrm{dist}$.
More precisely, we write
\begin{align*}
\sum_{\substack{ x_{l_{0}+1},\dots,x_{N} \in \tilde\Lambda  \\ \text{distinct}}}\ \prod_{ \{l,l' \}\in E(T)}\hspace{-0,4cm}J_{x_{l} x_{l'}} 
& \leq C^{|E(T)|} \ee^{-\frac{a}{2}L(T)}\hspace{-0,6cm}  \sum_{\substack{ x_{l_{0}+1},\dots,x_{N} \in \tilde\Lambda \\ \text{distinct}}}\
\prod_{ \{l,l' \}\in E(T)}\hspace{-0,2cm}\ee^{-\frac{a}{2} \mathrm{dist}(x_{l},x_{l'})} 
\leq   C^{|E(T)|} \ee^{-\frac{a}{2}L(T)},
\end{align*}
where in the last step we performed the sum over $x_{l_{0}+1},\dots,x_{N}$ by progressively stripping the tree from the leaves till the root
$x_{1}=i_{0}$. Since $L(T) \geq \mathrm{dist}(i_{0},j_{0})$, we obtain
\begin{equation*}
\begin{split}
\mathcal{A}_{i_{0},j_{0}} \leq \ee^{-\frac{a}{2} \mathrm{dist}(i_{0},j_{0})}  C\sum_{N \geq l_{0}} \frac{( C \beta m)^{N-1}}{|N-2|!} N^{N-2} \lesssim (\beta m)^{\ind{i_{0} \neq j_{0}}}\ee^{-\frac{a}{2} \mathrm{dist}(i_{0},j_{0})}.
\end{split}
\end{equation*}
where the small factor is gained only when $i_{0} \neq j_{0}$, since $l_{0} = 2$ in that case.
\medskip

\underline{\textit{Polynomially-decaying interaction:}} set $J_{ij}\! \lesssim \! \frac{1}{(1+|i-j|)^{a}}$, with $a\! >\! d.$

We have $\sup_{i\in \mathbb{Z}^{d}}\!  \sum_{j\in \mathbb{Z}^{d}}\!  J_{ij}\!  \lesssim\! 1$ $\forall a>d.$
Assume now $j_{0} \neq i_{0}.$ Wlog we can set $x_{1}=i_{0}$ and $x_{2}=j_{0}$. We denote by $p(T)$ the set of edges of the unique path on $T$ (on $\{1,\dots,N \}$) connecting $1$ to $2$. Let $v(T)$ be the set of vertices in $T$ belonging to such path and different from $1$ and $2$. We begin by stripping the tree of all the edges not on the path $p(T)$, and obtain
\begin{equation*}
\sum_{\substack{ x_{n_{0}+1},\dots,x_{N} \in \tilde\Lambda   \\ \text{distinct}}}\ \prod_{ \{l,l' \}\in E(T)} J_{x_l \,x_{l'}} \leq C^{N} \sum_{\{ x_{l}\}_{l \in v(T)}} \prod_{\{l,l'\}\in p(T)} J_{x_{l}\,x_{l'}}, 
\end{equation*}
where we collect a factor  $C$ for every summation $\sup _{y}\sum_{l \notin v(T)} J_{x_{l}y}$.
On the remaining edges along $p(T)$ we want to extract additional decay. For clarity, we now set $y_0:=i_0$, $y_{|v(T)|+1}:=j_0$. We then relabel the points in $v(T)$ by $1,\dots,|v(T)|$, so that the edges in the path $p(T)$ are now of the form $\{l,l+1 \}$ for $l = 0,\dots,|v(T)|\}$. For brevity, we also abridge $f(x):=\frac{1}{(1+|x|)^{a}}$ for $x \in \mathbb{Z}^{d}$. We can write
\begin{equation}\label{eq:pol-step-1}
\begin{split}
\sum_{\{ x_{l}\}_{l \in v(T)}} \prod_{\{l,l'\}\in p(T)} J_{x_{l},x_{l'}} & = \sum_{\{ y_{l}\}_{l=1,\dots,|v(T)|}} J_{y_{0}y_{1}} \prod_{l=1}^{|v(T)|} J_{y_{l} J_{l+1}}
\\
& \leq C^{N}\sum_{\{ y_{l}\}_{l=1,\dots,|v(T)|}} f(y_{0}-y_{1}) \prod_{l=1}^{|v(T)|} f(y_{l}-y_{l+1}) \;.
\end{split}
\end{equation}
Using $( 1+|x-z|)^{a}\lesssim_{a}  ( 1+|x-y|)^{a}+ (1+|y-z|)^{a},$
we argue
\begin{align*}
\sum_{y}f(x-y) f(y-z) %
&=\sum_{y}\frac{1}{  ( 1+|x-y|)^{a} (1+|y-z|)^{a}}\\
&=
\frac{1}{( 1+|x-z|)^{a}}\sum_{y}\frac{( 1+|x-z|)^{a} }{  ( 1+|x-y|)^{a} (1+|y-z|)^{a}}\\
&\lesssim_{a} \frac{1}{( 1+|x-z|)^{a}} \sum_{y}\frac{1 }{  (1+|y|)^{a}}\lesssim_{a} f(x-z).
\end{align*}

We now iteratively apply this bound on the right-hand side of \eqref{eq:pol-step-1} and obtain
\begin{equation*}
\sum_{\substack{ x_{n_{0}+1},\dots,x_{N} \in \tilde\Lambda   \\ \text{distinct}}}\ \prod_{ \{l,l' \}\in E(T)} J_{x_l \,x_{l'}} \leq C^{N} \sum_{\{ x_{l}\}_{l \in v(T)}} \prod_{\{l,l'\}\in p(T)} J_{x_{l}\,x_{l'}} \lesssim_{a} C^{N}  \frac{1}{(1+|i_{0}-j_{0}|)^{a}} .
\end{equation*}
The proof in the case $i_{0}=j_{0}$ is simpler since $p (T)$ is the empty set.  
The result now follows by replacing this estimate in \eqref{eq:Acalbound} and arguing as in the exponential case.
\end{proof}

\section{Norm estimates}
\label{sec:norm_estimate}

\subsection{Single-point model}
\label{sec:single-point}

In this section, we  remove reference to the spatial points in the set $\Lambda \subset \mathbb{Z}^{d}$ and consider computations
with the Grassmann variables $\{\cpsi_{\alpha}, \psi_{\alpha} \}_{\alpha =1,\dots,m}$ only. In particular, recall that
$\psi \cdot \psi = 2\sum_{\alpha = 1}^{m} \cpsi_{\alpha} \psi_{\alpha}$ and that $z = \sqrt{1+ \psi \cdot \psi}$,
compare with \eqref{eq:symmetric-product}  and above. 
With this notation, together with \eqref{eq:def-W-nu} and \eqref{eq:fact-notation},  the single-point partition function
\eqref{eq:single-site-normalisation},  is 
\begin{equation*}
\tilde{Z}_{\veps,m}  = \int \dd \nu =
 \int \dd \psi \ \frac{\ee^{-\varepsilon (z-1)}}{z}=
\bigg[\prod_{\alpha =1}^{m} \partial_{\cpsi_{\alpha}}\partial_{\psi_{\alpha}} \bigg] \ \frac{\ee^{-\varepsilon (z-1)}}{z}.
\end{equation*} 
The main result of this section is the following estimate. 
\begin{theorem}\label{one_point_estimate}
There exists a constant $C>1$ such that, for all $\veps \geq 0$ and all $m\in \mathbb{N}$,
\begin{equation}\label{one_point_one}
\max \left\{  \frac{ \left\| \ee^{-\varepsilon (z-1)}\right\|  }{\big | \tilde{Z}_{\veps,m} \big |}, \,
 \frac{ \left\| \frac{1}{z} \ee^{-\varepsilon (z-1)} \right\|  }{ \big |\tilde{Z}_{\veps,m} \big |}
\right \}\leq  C
\end{equation}
\end{theorem} 

Note that, by Lemma \ref{single_site_partition_function_value} below, $ \tilde{Z}_{\veps,m}>0$ $\forall \varepsilon\geq 0$ hence the ratios above are well-defined.
As explained in the following section, the above estimates are crucial for proving Proposition \ref{tree_estimate}.
In particular, although these bounds depend on $m$ and $\veps$, this dependence is optimal, as it can be controlled
by the single-site partition function $\tilde{Z}_{\veps,m}$. 
\begin{proof}
To prove the bounds we expand the numerators into polynomials in $\veps$ and use an explicit formula for the denominator.
Precisely we have
\begin{equation}\label{eq:expansion-epx}
 \ee^{-\varepsilon (z-1)}= \sum_{\ell=0}^m \frac{\veps^\ell}{\ell!}Q_\ell(\psi\cdot \psi)  ,\qquad
 \frac{\ee^{-\varepsilon (z-1)}}{z}= \sum_{\ell=0}^m \frac{\veps^\ell}{\ell!}P_\ell(\psi \cdot \psi).
\end{equation}  
where the functions $P_{\ell}$ and $Q_{\ell}$ are analytic in the complex open ball $B_{1} (0)$
\begin{equation}\label{eq:Taylor-coeff}
P_\ell(x)  =\frac{\left (1-\sqrt{1+x}\right )^\ell}{\sqrt{1+x}} =: \sum_{k=0}^{\infty}a_k(\ell)x^k,\qquad
Q_\ell(x) =\left (1-\sqrt{1+x}\right )^\ell =:\sum_{k=0}^{\infty}b_k(\ell)x^k .
\end{equation}

Hence,  by the triangle inequality
\begin{equation}\label{eq:norm-exp}
 \left \|  \frac{\ee^{-\varepsilon (z-1)}}{z}\right \|  \leq \sum_{\ell=0}^m\frac{\veps^\ell}{\ell!}\sum_{k=0}^m|a_k(\ell)| \|(\psi\cdot \psi)^k\|,\qquad \left \|  \ee^{-\varepsilon (z-1)}\right \| \leq \sum_{\ell=0}^m\frac{\veps^\ell}{\ell!}\sum_{k=0}^m|b_\ell(n)|  \|(\psi\cdot \psi)^k\| \;.
\end{equation}
where all sums stop at the order $m$ since $(\psi\cdot \psi)^{m+1}=0.$
We show in Lemma \ref{single_site_partition_function_value} below  the bound $|b_{k}(\ell)| \leq |a_{k}(\ell)|$, hence it suffices to control the right-hand side of the
first inequality in \eqref{eq:Taylor-coeff}.

Now, instead of using the crude bound $ \|(\psi\cdot \psi)^k\|\leq \|\psi\cdot \psi\|^k, $ we compute $\|(\psi\cdot \psi)^k\|$ \emph{exactly} by expanding 
   \[
    (\psi\cdot \psi)^k=\left(2\sum_{\alpha=1}^m\cpsi_{\alpha}\psi_{\alpha}\right)^k = 2^{k} \sum_{\alpha_{1},...,\alpha_{k}} \prod_{l=1}^{k} \cpsi_{\alpha_{l}}\psi_{\alpha_{l}}  \;.
   \]
Each of the $k$ factors have $m$ summands, and, due to the anticommutativity, only contributions with $k$ distinct $\alpha$'s survive. There are $\binom{m}{k}$ ways to choose such indices, and for each choice the factors can be permuted in $k!$ ways, yielding identical monomials, hence
   \[
        \|(\psi\cdot \psi)^k\|=2^k \frac{m!}{(m-k)!}  \;.
    \]
In Lemma \ref{single_site_partition_function_value} below we show the identity $ \tilde{Z}_{\veps,m}= 2^m m! \sum_{l=0}^{m}\frac{\veps^l}{l!} |a_m(l)|.$ Putting all this together and using $a_k(l)=0$ $\forall k<\ell,$ we obtain
\begin{equation*}
    \frac{\left \| \frac{\exp(-\veps(z-1))}{z} \right \|}{\big |\tilde{Z}_{\veps,m}\big |} \leq \frac{ \sum_{\ell=0}^m\frac{\veps^\ell}{\ell!}\sum_{k=\ell}^m \ |a_k(\ell)| \ 2^{k-m} \frac{1}{(m-k)!}}{\sum_{\ell=0}^m\frac{\veps^\ell}{\ell!} \ |a_m(\ell)|  }
 = \frac{\sum_{\ell=0}^m w_\ell R_\ell}{\sum_{\ell=0}^m w_\ell} \;,
\end{equation*}
where we introduced
\begin{equation}\label{def:R}
R_{\ell}:=\sum_{k=\ell}^m\frac{|a_k(\ell)|}{|a_m(\ell)|}\frac{2^{k-m}}{(m-k)!} \qquad  \qquad w_{\ell}:=\frac{\veps^\ell}{\ell!} \ |a_m(\ell)|.
\end{equation}
Note that $R_{\ell}$ is well-defined since $|a_m(\ell)|>0$ $\forall \ell=0,\dotsc ,m.$
The claim now follows by proving  $R_\ell\leq C$ for some constant $C>1.$ To this end, substituting the expressions \eqref{eq:coefficient} from Lemma \ref{coefficients} below for $|a_k(\ell)|$ and $|a_m(\ell)|$ into \eqref{def:R} yields
\begin{equation*}
R_\ell =\sum_{k=\ell}^m \frac{2^{(m-k)}}{(m-k)!}\frac{\binom{2k-\ell}{k-\ell}}{\binom{2m-\ell}{m-\ell}} =  \sum_{j=0}^{m-\ell}\frac{2^{j}}{j!}\frac{\binom{2m-\ell-2j}{m-\ell-j}}{\binom{2m-\ell}{m-\ell}},
\end{equation*}
where in the last step we performed the change of variables $m-k=j$. But  
\begin{align*}
        \frac{\binom{2m-\ell-2j}{m-\ell-j}}{\binom{2m-\ell}{m-\ell}}& \leq \frac{\binom{2m-\ell}{m-\ell-j}}{\binom{2m-\ell}{m-\ell}}=\frac{(m-\ell)!}{(m-\ell-j)!}\frac{m!}{(m+j)!}=\prod_{i=0}^{j-1}\frac{m-\ell-i}{m+i+1}\leq 1 \;,
\end{align*}
implying the uniform boundedness of $R_{\ell}$. Here we used the convention $\prod_{i\in \emptyset}:=1.$
\end{proof}

The next lemma gives an explicit formula for the coefficients $a_{k} (\ell), b_k(\ell)$  of the Taylor polynomials of $P_\ell$ and
$Q_\ell$. 
\begin{lemma}\label{coefficients}
With the notation of \eqref{eq:Taylor-coeff}, we have
\begin{gather}\label{eq:coefficient}
    a_k(\ell)=\ind{k \geq \ell}(-1)^{k}\ 2^{\ell-2k} \ \binom{2k-\ell}{k-\ell} \;,
    \qquad \quad 
    b_k(\ell)=  \ind{k \geq \ell} (-1)^{k}\ 2^{\ell-2k} \ \frac{\ell}{k}\binom{2k-\ell-1}{k-\ell} \;.
\end{gather}
Moreover it holds  $ |b_k(\ell)| \leq  |a_k(\ell)|$ $\forall k,\ell$.
\end{lemma} 
\begin{proof}
We prove the identity for $a_k(\ell)$; the corresponding claim for $b_k(\ell)$ can be proven with the same computations.

Since the constant coefficient of the Taylor polynomial of $\left (1-\sqrt{1+x}\right )$ vanishes,  we have $a_k(\ell)=0$  $\forall k<\ell.$ For $k\geq \ell$ our proof relies on the complex-integral representation of the derivatives.
In fact, $P_{\ell}(x)$ is analytic in the open unit ball $B_{1} (0),$ hence
\[
 a_k(\ell)=P_{\ell}^{(k)} (0)=\frac{1}{2\pi \ii}\oint_{\Gamma} \frac{P_\ell(x)}{x^{k+1}}\dd x=
\frac{1}{2\pi \ii}\oint_{\Gamma} 
\frac{\left (1-\sqrt{1+x}\right )^\ell}{\sqrt{1+x}} \frac{1}{x^{k+1}}\dd x,
\]
where $\Gamma$ is a counterclockwise complex contour around $0$, contained in $B_{1} (0).$ To simplify the integral
we perform the change of variables
\[
x=\phi(x'):=-\frac{4x'}{(1+x')^2} \ .
\]
This function is analytic in the open unit ball with derivative given by
\[
\phi'(x')=-4\frac{(1-x')}{(1+x')^3} \ .
\]
In particular, $\phi'(0)=-4\neq 0$ and hence there exists $0<r<1$ such that $\phi: B_{r} (0)\to B_{1} (0)$ is a biholomorphism
and  $\gamma:= \phi^{-1}(\Gamma)$ is also a counterclockwise complex contour around $0$, in $ B_{r} (0)$.
We compute $\sqrt{1+x}=\sqrt{1+\phi(x')} = (1-x')/(1+x')$  and
\begin{align*}
P_\ell(x)&=P_\ell(\phi(x'))=2^\ell\frac{(x')^\ell(1+x')^{1-\ell}}{1-x'}\\
  a_k(\ell) &=\frac{(-1)^k 2^{\ell-2k}}{2\pi \ii}\oint_{\gamma} \frac{1}{(x')^{1+k-\ell}}(1+x')^{2k-\ell} \dd x' \;.
\end{align*}
    The integral is evaluated using Cauchy theorem  for  $k\geq \ell,$ whence the  identity \eqref{eq:coefficient}
    for $a_{l} (\ell).$

Finally we show the inequality  $ |b_k(\ell)| \leq  |a_k(\ell)|$. For $k < \ell$ we have $ b_k(\ell)=a_k(\ell)=0 $. For $k \geq  \ell$ we argue
    \begin{align*}
        |b_k(\ell)|&= 2^{\ell-2k} \ \frac{\ell}{k}\binom{2k-\ell-1}{k-\ell} \leq 2^{\ell-2k} \  \frac{\ell}{k} \binom{2k-\ell}{k-\ell} \leq 2^{\ell-2k} \ \binom{2k-\ell}{k-\ell}=|a_k(\ell)|.
    \end{align*}
\end{proof}   
The next lemma gives an explicit sum representation of the single-site partition function.
\begin{lemma}\label{single_site_partition_function_value}
The following identity holds true $\forall \varepsilon\geq 0$
 \begin{equation}\label{eq:single_site_partition_function_value}
        \tilde{Z}_{\veps,m}= 2^m m! \sum_{\ell=0}^{m}\frac{\veps^\ell}{l!} |a_m(\ell)| \ 
    \end{equation}
\end{lemma}
where $a_{m} (\ell)$ are the coefficients introduced in \eqref{eq:coefficient}. In particular $  \tilde{Z}_{\veps,m}>0$
$\forall \varepsilon\geq 0$.
\begin{proof}
We expand $\frac{\exp(-\veps(z-1))}{z}$ in powers of $\veps$ as in \eqref{eq:expansion-epx}; the Grassmann integral projects the polynomial $P_{\ell}(\psi \cdot \psi)$ on the top monomial of degree $2m$, so that 
\[
\tilde{Z}_{\veps,m} =\sum_{\ell=0}^m \frac{\veps^\ell}{\ell!}
\int \dd \psi \ 
P_{\ell}(\psi \cdot \psi) = \sum_{\ell=0}^m \frac{\veps^\ell}{\ell!} a_m(\ell) \int \dd \psi \ 
(\psi \cdot \psi)^m = \sum_{\ell=0}^m \frac{\veps^\ell}{\ell!} a_m(\ell) (-1)^m m! \;,
\]
where in the last step we used the anticommutativity of the Grassmann variables. Finally using that $a_m(\ell) (-1)^m=|a_m(\ell)|$, we get the identity \eqref{eq:single_site_partition_function_value}.

The bound  $  \tilde{Z}_{\veps,m}>0$
$\forall \varepsilon\geq 0$ follows from $|a_{m} (0) |= 2^{-2m} \ \binom{2m}{m} >0.$
\end{proof} 

\subsection{Proof of the optimal activity estimates}
\label{sec:optimal-activity}

This section is devoted to proving Proposition \ref{tree_estimate}.

Recall the definition of $K(Y)$ and of $K_{i_{0},j_{0}}(Y)$ from \eqref{eq:pol-external-field} and of the connected component
 $ (\ee^{-\cw(Y)})_{\mathrm{conn}}$ from \eqref{eq:expansion0} and  \eqref{eq:expansion}.
While the recusive  definition of $ (\ee^{-\cw(Y)})_{\mathrm{conn}}$  is conceptually natural,
it is not convenient for explicit computations.
We shall instead rely on the following well-known representation. 
\begin{theorem}[Battle--Brydges--Federbusch Formula]\label{thm:BBF-formula}
With the notation set above
\begin{equation*}
(\ee^{-\cw(Y)})_{\mathrm{conn}} = \sum_{T \; \mathrm{on} \; Y} \int \dd p^{T}(s) \bigg[\prod_{\{i, j\}\in E(T)}-\cw_{ij}\bigg]
\exp\bigg(-\sum_{\{i,j \} \in E_{Y} } s_{ij}\cw_{ij}\bigg) \;,
\end{equation*}
where the sum is over the spanning trees $T$ on $Y$, $E(T)$ denotes the set of the edges of $T$, and $\dd p^{T}$ is some probability measure on $ [0,1]^{E_{Y}}$. 
\end{theorem} 
This expansion was first introduced by Battle and Federbush \cite{Battle1984}, see \cite{Brydges:1984vu} for a comprehensive review. Although the probability measure $\dd p^{T}$ is explicitly known, its precise form will not be needed here.
Alternative versions for this expansion were subsequently proposed by Brydges--Kennedy \cite{Brydges1987} and by Abdesselam--Rivasseau \cite{Abdesselam1995}. Note in passing that, although the setting of \cite{Brydges1987} is formulated for complex-valued Gibbs weight, the expansion applies more generally to commutative Banach algebras (which is our case, see also Lemma \ref{lemma:norm}), as clarified in \cite{Abdesselam1995}.

Using the representation of Theorem \ref{thm:BBF-formula} for $(\ee^{-\cw(Y)})_{\mathrm{conn}} $ we write
\begin{align*}
\big|K (Y)\big|&=\bigg|\sum_{T \; \mathrm{on} \; Y} \int \dd p^{T}(s)
\frac{1}{\tilde{Z}^{|Y|}_{\veps,m}} \int \dd \psi_{Y}
\bigg[\prod_{\{l,l'\}\in E(T)}\hspace{-0.4cm}-\cw_{ll'}\bigg]  \exp\bigg(-\hspace{-0.4cm}\sum_{ \{l,l'\} \in E_Y}\hspace{-0.4cm} s_{ll'}\cw_{ll'}\bigg) \cv^{Y}  \bigg| \\
&\leq \sum_{T \; \mathrm{on} \; Y} \int \dd p^{T}(s)
\frac{1}{\tilde{Z}_{\veps,m}^{|Y|}} \bigg \| \bigg[\prod_{\{l, l'\}\in E(T)}\hspace{-0.4cm}-\cw_{ll'}\bigg] \exp\bigg(-\hspace{-0.4cm}\sum_{\{l,l'\} \in E_Y }\hspace{-0.4cm} s_{ll'}\cw_{ll'}\bigg) \cv^{Y} \bigg \| \\
&\leq \sum_{T \; \mathrm{on} \; Y}  \frac{1}{\tilde{Z}^{|Y|}_{\veps,m}}
\sup_{s\in[0,1]^{E(T)}}\bigg \|\bigg[ \prod_{\{l, l'\}\in E(T)}\hspace{-0.4cm}-\cw_{ll'}\bigg] \exp\bigg(-\hspace{-0.4cm}\sum_{\{l,l'\} \in E_Y }\hspace{-0.4cm} s_{ll'}\cw_{ll'}\bigg) \cv^{Y} \bigg \| 
\end{align*}
where, in the second line, we applied Lemma \ref{lemma:norm} and in the third line we applied the fact that  $\dd p^{T}$ is a probability measure over $[0,1]^{E(T)}.$ 
Since  $\| \cpsi_{i_{0},\alpha} \psi_{j_{0},\alpha}\|=1,$ the same bound holds for $|K_{i_{0},j_{0}}(Y)|.$

We claim that for any spanning tree $T$ on $Y$ the following bound holds true 
\begin{equation}\label{eq:single_tree_estimate}
\hspace{-0.2cm}\frac{1}{\tilde{Z}^{|Y|}_{\veps,m}}\hspace{-0.1cm} \sup_{s\in[0,1]^{E(T)}}\bigg \| \bigg[\prod_{\{l, l'\}\in E(T)}\hspace{-0.4cm}-\cw_{ll'}\bigg]
\exp\bigg(-\hspace{-0.4cm}\sum_{\{l,l'\} \in E_Y }\hspace{-0.4cm} s_{ll'}\cw_{ll'}\bigg) \cv^{Y} \bigg \|
\leq C^{|Y|}(\beta m)^{|Y|-1}\hspace{-0.4cm}\prod_{ \{l,l' \}\in E(T)}\hspace{-0.2cm} J_{ll'} \;,
\end{equation}
which implies the result. 
The rest of this section is devoted to the proof of the above bound. Note that while $\|z\|\sim m^{m}$ we have
$\|z^{2}\|=1+2 m\leq \| z\|^{2} \sim m^{2m} $ see \eqref{eq:even_z_bound}. The key idea is then to expand
\begin{equation}\label{tree_integrand}
\bigg[\prod_{\{l, l'\}\in E(T)}\hspace{-0.4cm}-\cw_{ll'}\bigg]
\exp\bigg(-\hspace{-0.4cm}\sum_{\{l,l'\} \in E_Y}\hspace{-0.4cm} s_{ll'}\cw_{ll'}\bigg) \cv^{Y} 
\end{equation}
keeping track of the even powers of the $z$'s. Recall the definitions of $\cw_{ll'}$ and $\cv^{Y}$ from \eqref{eq:def-W-nu}.
Setting $A_{ll'}:=- \beta (1+\psi_l \cdot \psi_{l'}) $ and $G_{T}:=\exp \Big(- \sum_{\{l,l'\} \in E_Y} s_{ll'} J_{ll'}A_{ll'}\Big)
$  we write 
\begin{align*}
\eqref{tree_integrand}
&= \prod_{\{l, l'\}\in E(T)}\hspace{-0.4cm}- J_{ll'}(A_{ll'}+\beta z_{l}z_{l'})\ G_{T}\ \prod_{\{l,l'\} \in E_Y} \ee^{-s_{ll'} J_{ll'}\beta z_{l}z_{l'}}
\prod_{j\in Y}  \frac{\ee^{-\varepsilon (z_{j}-1)}}{z_{j}}.
\end{align*}
We expand further $\ee^{-s_{ll'} J_{ll'}\beta z_{l}z_{l'}}$ separating even and odd powers: 
\begin{align*}
 \ee^{-s_{ll'} J_{ll'}\beta z_{l}z_{l'}}&= \sum_{k\geq 0}\frac{(s_{ll'}\beta J_{ll'})^{2k} z_l^{2k}z_{l'}^{2k}}{(2k)!} -s_{ll'}\beta J_{ll'} z_lz_{l'}
  \sum_{k\geq 0}\frac{(s_{ll'}\beta J_{ll'})^{2k} z_l^{2k}z_{l'}^{2k}}{(2k+1)!} \\
&=
 1+ B_{ll'}^{0}+ \beta z_lz_{l'} B_{ll'}^{1},
\end{align*}
where $ B_{ll'}^{0}, B_{ll'}^{1}$ are polynomials in $z_{l}^{2}z_{l'}^{2}$
\[
 B_{ll'}^{0}:= \sum_{k\geq 1}\frac{(s_{ll'}\beta J_{ll'})^{2k} z_l^{2k}z_{l'}^{2k}}{(2k)!},
 \qquad  B_{ll'}^{1}:= -s_{ll'}J_{ll'}\sum_{k\geq 0}\frac{(s_{ll'}\beta J_{ll'})^{2k} z_l^{2k}z_{l'}^{2k}}{(2k+1)!}.
\]
Note that, in view of Lemma \ref{lemma:spurious-terms} below, it is convenient to keep one $\beta$ factor attached to $z_{l}z_{l'}.$
The next lemma gives a bound on the norm of the four functions $A_{ll'},G_{T},B_{ll'}^{0},B_{ll'}^{1}.$
\begin{lemma}\label{lemma:simple-estimates}
Assume $\sup_{l \in \mathbb{Z}^{d}}\sum_{l' \in \mathbb{Z}^{d}}J_{ll'} <\infty$. Then, for any tree $T$ on $Y$, any $\{l,l'\} \in E(T)$ and any $s \in [0,1]^{E(T)}$ we have
\begin{align*}
\|A_{ll'} \| &\leq  3m \beta \;, &  \| G_{T}\| \leq \ee^{|Y|  C m \beta} \;,
\\
\|B^{0}_{ll'} \| & \leq  3m \beta J_{ll'} \ee^{C m\beta}\;, 
&
\|B^{1}_{ll'} \| \leq  J_{ll'} \ee^{C m \beta}\;,
\end{align*}
for some universal constant $C>1$ depending on $J.$ 
\end{lemma}
\begin{proof}
The  bound on $\|A_{ll'} \|$ follows from $\|\psi_l \cdot \psi_{l'} \|=2m$ $\forall l,l'\in  \Lambda$. To prove the bound on
$ \| G_{T}\|$  note that
\begin{equation*}
\bigg\|\sum_{l\neq l'\in Y } s_{ll'} J_{ll'}A_{ll'}\bigg\|
\leq |Y| \sup_{l \in Y} \sum_{l'\in Y} J_{ll'} \| A_{ll'}\| \;,
\end{equation*}
whence $\| G_{T}\|\leq \ee^{\|\sum_{l\neq l'\in Y } s_{ll'} J_{ll'}A_{ll'}\|}\leq \ee^{C|Y|\beta m}.$
The bounds on $\|B^{0}_{ll'} \|$ and  $\|B^{1}_{ll'} \|$ follow  using $\|z_{l}^{2}\|=1+2m\leq 3m$ and
bounding the series by an exponential function. 
\end{proof}
Putting these results together we obtain 
\begin{align}\nonumber
&\|\eqref{tree_integrand}\|\leq \hspace{-0.4cm} \prod_{\{l, l'\}\in E(T)}\hspace{-0.4cm} J_{ll'}  \| G_{T}\|
\left\|
\prod_{\{l, l'\}\in E(T)}\hspace{-0.4cm}(A_{ll'}+\beta z_{l}z_{l'})\hspace{-0.2cm} \prod_{\{l,l'\} \in E_Y}\hspace{-0.2cm}  ( 1+ B_{ll'}^{0}+ \beta z_lz_{l'} B_{ll'}^{1} )
\prod_{j\in Y}  \frac{e^{-\varepsilon (z_{j}-1)}}{z_{j}} \right\|\\
&\quad \leq \ \ee^{C|Y|\beta m} \hspace{-0.4cm} \prod_{\{l, l'\}\in E(T)}\hspace{-0.4cm} J_{ll'}
\left\|
\prod_{\{l, l'\}\in E(T)}\hspace{-0.4cm}(A_{ll'}+\beta z_{l}z_{l'})\hspace{-0.2cm} \prod_{\{l,l'\} \in E_Y} \hspace{-0.2cm}( 1+ B_{ll'}^{0}+ \beta z_lz_{l'} B_{ll'}^{1} )
\prod_{j\in Y}  \frac{\ee^{-\varepsilon (z_{j}-1)}}{z_{j}} \right\|.\label{eq:inter1}
\end{align}
The product
\begin{equation}\label{eq:prod-z}
\prod_{\{l, l'\}\in E(T)}\hspace{-0.4cm}(A_{ll'}+\beta z_{l}z_{l'}) \prod_{\{l,l'\} \in E_Y} ( 1+ B_{ll'}^{0}+ \beta z_lz_{l'} B_{ll'}^{1} )
\prod_{j\in Y}  \frac{\ee^{-\varepsilon (z_{j}-1)}}{z_{j}} 
\end{equation}
still contains contributions from $z_{l}z_{l'},$ whose norm is big. To 
avoid suboptimal bounds in the norm, we extract the spurious factors $z_{l}z_{l'}$ as follows
\begin{align*}
&\prod_{\{l, l'\}\in E(T)}\hspace{-0.4cm}(A_{ll'}+\beta z_{l}z_{l'}) \prod_{\{l,l'\} \in E_Y} (1+  B_{ll'}^{0}+ \beta z_lz_{l'} B_{ll'}^{1} )=\\
&\qquad \sum_{\substack{I \subset E(T) \\  \tilde{I} \subset E_Y}} \bigg( \prod_{\{ l,l'\} \in I} A_{ll'} \bigg) \bigg( \prod_{\{ l,l'\} \in \tilde I} (1+B_{ll'}^{0}) \bigg)  \bigg(  \prod_{\{ l,l'\} \in I^{c}}  \beta  z_{l}z_{l'}\bigg)
\bigg(  \prod_{\{ l,l'\} \in \tilde I^{c}} B_{ll'}^{1} \beta  z_{l}z_{l'}\bigg) \;,
\end{align*}
where we introduced the complements $I^{c}:= E(T)\setminus I$ and $\tilde I^{c}:= E_Y \setminus \tilde I$.
By Lemma \ref{lemma:simple-estimates}, we have 
\begin{align}\label{eq:exp-step-2}
\|\eqref{eq:prod-z}\|&\leq \sum_{\substack{I \subset E(T) \\  \tilde{I} \subset E_Y}} (3m \beta)^{I} \prod_{\{ l,l'\} \in \tilde I} (1+3m\beta J_{ll'} \ee^{Cm\beta}) \  
  \prod_{\{ l,l'\} \in \tilde I^{c}} J_{ll'} \ee^{C m\beta}\ \cdot\\
&\qquad \qquad 
\ \bigg\|\prod_{\{ l,l'\} \in I^{c}}  \beta  z_{l}z_{l'} \  
\prod_{\{ l,l'\} \in \tilde I^{c}}  \beta z_{l}z_{l'}\ \prod_{j\in Y}  \frac{\ee^{-\varepsilon (z_{j}-1)}}{z_{j}}  \bigg\|.
\nonumber
\end{align}
We can estimate the contribution of the spurious terms by using, when necessary, the factor $1/z_{i}$ from the density $\ee^{-\veps(z_{i}-1)}/z_{i}$ and by means of Theorem \ref{one_point_estimate}.
\begin{lemma}\label{lemma:spurious-terms}
Let $T$ be a spanning tree on $Y$. For any $I \subset E(T)$ and $\tilde I \subset E_Y$, we have
\begin{equation*}
\bigg\|\bigg(  \prod_{\{ l,l'\} \in I^{c}}  \beta  z_{l}z_{l'}\bigg)
\bigg(  \prod_{\{ l,l'\} \in \tilde I^{c}}  \beta z_{l}z_{l'}\bigg) \prod_{j\in Y}  \frac{\ee^{-\varepsilon (z_{j}-1)}}{z_{j}} \bigg\| \leq C^{|Y|} \tilde Z_{\veps,m}^{|Y|} \big(3m\beta \big)^{|I^{c}| + |\tilde I^{c}|} \;.
\end{equation*}
for some universal constant $C$.
\end{lemma}
\begin{proof}
For any choice of $I$ and $\tilde I$, there are unique integers $(q_{i})_{i \in Y}$, $q_{i}=q_{i}(I,\tilde I) \in \mathbb{N}_{\geq 0}$, such that
\begin{equation}\label{eq:z_even_odd}
\bigg(  \prod_{\{ l,l'\} \in I^{c}}  \beta  z_{l}z_{l'}\bigg)
\bigg(  \prod_{\{ l,l'\} \in \tilde I^{c}}  \beta z_{l}z_{l'}\bigg)=\prod_{i\in Y}  \beta^{q_{i}/2}  z_i^{q_i} \;.
\end{equation}
Even powers in \eqref{eq:z_even_odd} are estimated using $\|z^{2}\|=1+2m$, whereas the factors $1/z_{i}$ are used to compensate odd powers. We obtain
\begin{equation*}
\begin{split}
\bigg\|\bigg( & \prod_{\{ l,l'\} \in I^{c}}  \beta  z_{l}z_{l'}\bigg)
\bigg(  \prod_{\{ l,l'\} \in \tilde I^{c}}  \beta z_{l}z_{l'}\bigg) \prod_{j\in Y} \frac{\ee^{-\varepsilon (z_{j}-1)}}{z_{j}} \bigg\|
\\
&
\leq 
\bigg( \prod_{\substack{i\in Y \\ q_i \text{ even} }}\beta^{q_i/2}\left \|z_i^{q_i}\right \| \|\frac{\ee^{-\varepsilon (z_{i}-1)}}{z_{i}} \|\bigg) \bigg( \prod_{\substack{i\in Y \\ q_i \text{ odd} }}\beta^{q_i/2}\|z_i^{q_i-1}\| \| z_{i} \frac{\ee^{-\varepsilon (z_{i}-1)}}{z_{i}} \| \bigg)
\\
& \leq C^{|Y|} \tilde Z_{\veps,m}^{|Y|}  \bigg( \prod_{\substack{i\in Y  }}\beta^{q_i/2} (1+2m)^{q_{i}/2}\bigg) 
 = C^{|Y|} \tilde Z_{\veps,m}^{|Y|} \big(3m\beta \big)^{|I^{c}| + |\tilde I^{c}|} 
\end{split}
\end{equation*}
where in the last inequality we used Theorem \ref{one_point_estimate}, while in the last identity $1+2m \leq 3m$ and $\sum_{i}q_{i} = 2|I^{c}| + 2|\tilde I^{c}|$.
\end{proof}
Using this lemma together with \eqref{eq:exp-step-2} we obtain
\begin{align*}
&\bigg\|  \prod_{\{l,l' \}  \in E(T)} (A_{ll'}+ \beta z_{l} z_{l'}) \  
 \prod_{\{l,l' \} \in E_Y}(1+B_{ll'}^{0}+B_{ll'}^{1}  \beta  z_{l} z_{l'}) \   \prod_{j\in Y} \frac{\ee^{-\varepsilon (z_{j}-1)}}{z_{j}} \bigg\|\\
&\qquad\qquad \leq  C^{|Y|} \tilde Z_{\veps,m}^{|Y|}   (3m \beta)^{|E(T)|}  \prod_{\{ l,l'\} \in E_Y} (1+6m\beta J_{ll'} \ee^{Cm\beta})
\\
&\qquad \qquad \leq C^{|Y|} \tilde Z_{\veps,m}^{|Y|}   (3m \beta)^{|E(T)|} \ee^{|Y|Cm\beta \ee^{C m \beta}}
\leq C^{|Y|} \tilde Z_{\veps,m}^{|Y|}   (3m \beta)^{|E(T)|}  \;,
\end{align*}

where in the last step we used once more that $\sup_{l}\sum_{l'}J_{ll'} <\infty$ holds uniformly and we used the condition $\beta m\leq 1$.
Plugging this bound into \eqref{eq:inter1} completes the proof of \eqref{eq:single_tree_estimate}, and hence of Proposition \ref{tree_estimate}.

\appendix

\section{Proof of the high-temperature bound}\label{appendix_expansion}

In this section, we provide the details of the proof of  Proposition \ref{expansion_lemma}. 
Recall the identity $\big[ \bar{\psi}_{i,\alpha} \psi_{j,\alpha}\big]^{\Lambda}_{\beta,\veps,m}= \partial_{\rho_{i,\alpha}}\partial_{ \bar{\rho}_{j,\alpha}}\ln \tilde{Z}^{\Lambda}_{\beta,\veps,m}(\rho) \big|_{\rho = 0},$ see \eqref{eq:def2point-function}.
We will start by constructing a polymer expansion for $\ln \tilde{Z}^{\Lambda}_{\beta,\veps,m}(\rho)$ and hence for
$\big[ \bar{\psi}_{i,\alpha} \psi_{j,\alpha}\big]^{\Lambda}_{\beta,\veps,m}$ in Lemma \ref{le:polreplog} and
Corollary \ref{cor:polrep2p} below. These are crucial for the proof of Proposition  \ref{expansion_lemma} which is given at the end of the section.

Our presentation follows closely \cite{Brydges:1984vu}, adapting the computation for the partition function to the two-point function.
Recall the definition of $\mathcal{K}(Y)$ in \eqref{eq:defKgener}.
The basic idea is that the polymer gas representation of Lemma \ref{lemma:high-temperature}
\begin{equation*}
\frac{Z^{\Lambda}_{\beta,\veps,m}(\rho)}{\prod_{j \in \Lambda} \tilde{Z}_{\veps,m}^{\{ j\}}(\rho)} = 1+ \sum_{N \geq 1}\frac{1}{N!}\sum_{\substack{Y_1,\ldots,Y_N\subset \Lambda\\ |Y_l|>1}} \Bigg[\prod_{l=1}^N \mathcal{K}(Y_l) \Bigg] \varphi(Y_1,\ldots,Y_N) \;,
\end{equation*}
is the grand canonical generating function of a polymer system with activities $\mathcal{K}(Y_{\ell})$, interacting with a hard-core potential
$\varphi$. Without this potential, the sum above factors yielding
\[
\log Z^{\Lambda}_{\beta,\veps,m}(\rho)  =
\sum_{j \in \Lambda} \log \tilde{Z}_{\veps,m}^{\{ j\}}(\rho)+\sum_{\substack{Y\subset \Lambda\\ |Y|>1}}\mathcal{K}(Y).
\] 
To overcome the lack of factorisation due to $\varphi$, one notices that $\varphi$ is a Gibbs weight so that it can be expanded into connected
components as in \eqref{eq:expansion}. The resulting formula for the logarithm is the content of the next lemma.
\begin{lemma}\label{le:polreplog}
The following identity holds
\[
\log Z^{\Lambda}_{\beta,\veps,m}(\rho)  =
\sum_{j \in \Lambda} \log \tilde{Z}_{\veps,m}^{\{ j\}}(\rho)
 +\sum_{N\geq 1}\frac{1}{N!}\sum_{\substack{Y_1,\ldots,Y_{N}\subset \Lambda\\ |Y_l|>1}}\Bigg[\prod_{l=1}^{N}\mathcal{K}(Y_l)\Bigg]
 \Big (\varphi(Y_1,\ldots,Y_N)\Big)_{\mathrm{conn}}\; ,
\]
provided the sum over $N$ in the right-hand side is absolutely convergent.
Here, analogously to \eqref{eq:expansion}, $\big(\varphi(Y_1,\ldots,Y_N)\big)_{\mathrm{conn}}$ is defined by setting  $\big(\varphi(Y)\big)_{\mathrm{conn}}=1$ and recursively for $N \geq 2$
\[
        \varphi(Y_1,\dots,Y_N)=\sum_{\Pi \; \mathrm{part} \; \{1,\dots,N\}}\prod_{I\in \Pi}\Big(\varphi\big((Y_{l})_{l \in I}\big)\Big)_{\mathrm{conn}} \;.
    \]
\end{lemma}
\begin{proof}
We plug the expansion for $\varphi$ into the polymer expansion of Lemma \ref{lemma:high-temperature} and rearrange the summation to obtain
\[
    \frac{Z^{\Lambda}_{\beta,\veps,m}(\rho)}{\prod_{j \in \Lambda} \tilde{Z}_{\veps,m}^{\{ j\}}(\rho)}
        =1+\sum_{N\geq 1}\frac{1}{N!}\sum_{\Pi \; \mathrm{part} \; \{1,\dots,N\}}\prod_{I\in \Pi}\left(\sum_{\substack{Y_1,\ldots,Y_{|I|}\subset \Lambda\\ |Y_l|>1}} \Bigg[\prod_{l=1}^{|I|} \mathcal{K}(Y_l) \Bigg] \Big(\varphi(Y_1,\ldots,Y_{|I|})\Big)_{\mathrm{conn}}\right) \;,
\]
where the function in the brackets depends on $I$ via its cardinality $|I|$ only and takes values in the even elements of the Grassmann algebra
generated by $\{ \bar{\rho}_{j,\alpha},\rho_{j,\alpha} \}_{j \in \Lambda, \alpha \in \{1,\dots,m \}}$.
The claim now follows from the following general combinatorial identity, valid for any function $f$ taking values in a commutative Banach algebra:
\begin{equation}\label{eq:claim-exponential}
\sum_{N\geq 1} \frac{1}{N!}\sum_{\Pi \; \mathrm{part} \; \{1,\dots,N \}}\prod_{I\in \Pi}f(|I|) = \sum_{M\geq1}\frac{1}{M!}\left(\sum_{n=1}^{\infty}\frac{1}{n!} f(n)\right)^M.
\end{equation}
Note that the sum on right-hand side is absolutely convergent by assumption.
To prove the identity we  first rearrange the sum according to the number of subsets in the partition and then sum over the size of each such set.
We obtain
\begin{equation*}  
 \sum_{N\geq 1}\frac{1}{N!}\sum_{M=1}^{N}\frac{1}{M!}\sum_{\substack{n_1,\ldots,n_M\geq1\\\sum_ln_l=N}}\frac{N!}{\prod_l n_l!}\prod_{l=1}^{M}f(n_l)
    =\sum_{M\geq 1}\frac{1}{M!}\sum_{n_1,\ldots,n_M\geq 1}\frac{1}{\prod_l n_l!}\prod_{l=1}^{M}f(n_l) \;.
\end{equation*}
The right-hand side leads directly to the exponential in \eqref{eq:claim-exponential}.
\end{proof}
As a direct consequence, we obtain the following representation for the two-point function. 
\begin{corollary}\label{cor:polrep2p}
With the notation of Proposition \ref{expansion_lemma}, we have
\begin{align}\label{derivative_sum}
 \big[ \bar{\psi}_{i_{0},\alpha} \psi_{j_{0},\alpha}\big]^{\Lambda}_{\beta,\veps,m}  &= \ind{i_{0}=j_{0}}\  K_{i_{0}j_{0}}(\{i_{0}\})\\
 & + 
 \sum_{N=1}^{\infty}\frac{1}{(N-1)!}\sum_{\substack{Y_{1} \ni i_0,j_0\\Y_2,\ldots,Y_N\subset \Lambda}, |Y_{l}|>1}\hspace{-0,6cm}K_{i_{0}j_{0}}(Y_1)
\left[ \prod_{l=2}^{N}K(Y_l)\right]\, \Big(\varphi(Y_1,\ldots,Y_N)\Big)_{\mathrm{conn}},\nonumber
\end{align}
where for $N=1$ the product is meant to be empty and equal to $1$.
\end{corollary}
\begin{proof}
We plug Lemma \ref{le:polreplog} in the identity \eqref{eq:def2point-function}.
The first term in \eqref{derivative_sum} comes from $\partial_{\rho_{i_{0},\alpha}} \partial_{\bar{\rho}_{j_{0},\alpha}}\log \tilde{Z}_{\veps,m}^{\{ i_{0}\}}(\rho) \big|_{\rho = 0} $, whereas the series comes from differentiating $\prod_{l=1}^{N} \mathcal{K}(Y_{l})$.
\end{proof}
To estimate the right-hand side of \eqref{derivative_sum}  we will also need the following bound on the connected component of the hard-core interaction.
\begin{proposition}\label{lemma_bound_phi}
    \[
        |(\varphi(Y_1,\cdots,Y_N))_{\mathrm{conn}}|\leq \sum_{T \; \mathrm{on} \; \{1,\dots,N \}}\prod_{\{l,l'\}\in E(T)}\ind{Y_l\cap Y_{l' } \neq \varnothing} \;.
    \]
\end{proposition}
\begin{proof}
This bound is well-known in statistical mechanics. A straightforward proof can be obtained applying the forest formula \cite{Abdesselam1995}.
See also \cite{Brydges:1984vu}.
\end{proof}

We are finally in a position to prove Proposition \ref{expansion_lemma}.
\begin{proof}[Proof of Proposition \ref{expansion_lemma}]
Inserting absolute values in \eqref{derivative_sum} we get
\[
| \big[ \bar{\psi}_{i_{0},\alpha} \psi_{j_{0},\alpha}\big]^{\Lambda}_{\beta,\veps,m}|\leq \sum_{N\geq 0}B (N),
\]
with $\max_{N=0,1} B(N)\leq \sum_{\substack{Y_1\subset\Lambda\\Y_{1} \ni i_0,j_0}}|K_{i_{0}j_{0}}(Y_{1})|.$
For $N\geq 2$ we argue, using  Proposition \ref{lemma_bound_phi},
\begin{align*}
&B (N)= \tfrac{1}{(N-1)!}\hspace{-0,4cm}\sum_{\substack{Y_{1} \ni i_0,j_0\\Y_2,\ldots,Y_N\subset \Lambda}, |Y_{l}|>1}\hspace{-0,6cm}  |K_{i_{0}j_{0}}(Y_1)|
\left[ \prod_{l=2}^{N}|K(Y_l)|\right]\,  \Big|\Big(\varphi(Y_1,\ldots,Y_N)\Big)_{\mathrm{conn}}\Big|\\
&\quad \leq \tfrac{1}{(N-1)!}  \sum_{T \; \mathrm{on} \; \{1,\dots,N \}} \sum_{\substack{Y_{1} \ni i_0,j_0\\Y_2,\ldots,Y_N\subset \Lambda}, |Y_{l}|>1}\hspace{-0,6cm}  |K_{i_{0}j_{0}}(Y_1)|
\left[ \prod_{l=2}^{N}|K(Y_l)|\right] \prod_{\{q,q'\}\in E(T)}\ind{Y_q\cap Y_{q' } \neq \varnothing}\\
&\quad =\tfrac{1}{(N-1)!} \sum_{\{d_l\}_{l=1}^{N}}\sum_{T \text{ with } \{d_l\}_{l}}\ \sum_{\substack{ Y_1 \ni i_{0},j_{0}\\Y_2,\ldots,Y_N\subset \Lambda},|Y_{l}|>1}
\hspace{-0,3cm} |K_{i_{0},j_{0}}(Y_1)| \left[\prod_{l=2}^{N}|K(Y_l)|\right]\prod_{\{q,q'\}\in E(T)}\ind{Y_{q}\cap Y_{q'}\neq \varnothing},
\end{align*}
where in the last line we decomposed the sum over trees into a sum over incident numbers and sum over trees with a specific incidence number.
We fix now the vertex $1$ as the root of the tree, which induces a partial order on the vertices of the tree as follows:
\[
 b<a \quad \text{ if } b\in \mathrm{path}(a,1),
\]
where $\mathrm{path}(a,1)$ denotes the unique path from $1$ to $a$. We say that $a$ is a descendant of $b$ if $b<a$; in particular, $a$ is a child
of $b$ if $b<a$ and $\{a,b\}$ is an edge of the tree.
The sum over $Y_{1},\dots,Y_{N}$ is estimated by progressively stripping the tree from the leaves (the vertices with no descendants) towards the root. Let $a$ be a leaf with parent $b$. Focusing on the summation over $Y_{a}$, we obtain
\[
 \sum_{Y_a}|K(Y_a)|\ind{Y_{a}\cap Y_b\neq \varnothing}  
    \leq \sum_{k\in Y_b}\sum_{\substack{Y_a \ni k}}|K(Y_a)| \leq |Y_b|\ \sup_{k\in \Lambda }\sum_{\substack{Y_a \ni k}}|K(Y_a)||Y_a|^{d_a-1} \;,
\]
since, in fact, $d_{a}=1$ for a leaf. Thus, each child of a vertex $b$ produces a factor $|Y_b|$. Proceeding inductively as above yields
\begin{multline*}
\sum_{\substack{ Y_1 \ni i_{0},j_{0}\\Y_2,\ldots,Y_N\subset \Lambda},|Y_{l}|>1}
\hspace{-0,6cm} |K_{i_{0},j_{0}}(Y_1)| \left[\prod_{l=2}^{N}|K(Y_l)|\right]\prod_{\{q,q'\}\in E(T)}\ind{Y_{q}\cap Y_{q'}\neq \varnothing}
\\
\leq \Bigg [ \sum_{\substack{Y_1 \ni i_{0},j_{0}}}|K_{i_{0}j_{0}}(Y_1)||Y_1|^{d_1} \Bigg]\ \prod_{l=2}^{N}\Bigg[\sup_{k\in \Lambda }\sum_{\substack{Y_l \ni k}}|K(Y_l)||Y_l|^{d_l-1}\Bigg].
\end{multline*}
Since the right-hand side depends only on the incidence numbers of the tree, we can bound the sum over trees with fixed incidence numbers $\{d_l\}_{l}$ by Cayley's Theorem and obtain
\begin{align*}
B (N)&\leq \tfrac{1}{(N-1)!} \sum_{\{d_l\}_{l=1}^{N}}  \tfrac{(N-2)!}{\prod_{l}(d_l-1)!}
\Bigg [ \sum_{\substack{Y_1 \ni i_{0},j_{0}}}|K_{i_{0}j_{0}}(Y_1)||Y_1|^{d_1} \Bigg]\ \prod_{l=2}^{N}\Bigg[\sup_{k\in \Lambda }\sum_{\substack{Y_l \ni k}}|K(Y_l)||Y_l|^{d_l-1}\Bigg]\\
&\leq  \Bigg [ \sum_{\substack{Y_1 \ni i_{0},j_{0}}}|K_{i_{0}j_{0}}(Y_1)|\sum_{d_1}\tfrac{|Y_1|^{d_1}}{(d_1-1)!} \Bigg]\
 \prod_{l=2}^{N}\Bigg[\sup_{k\in \Lambda }\sum_{\substack{Y_l \ni k}}|K(Y_l)|\sum_{d_l}\tfrac{|Y_l|^{d_l-1}}{(d_l-1)!}\Bigg]\\
&\leq  \Bigg [ \sum_{\substack{Y_1 \ni i_{0},j_{0}}}|K_{i_{0}j_{0}}(Y_1)|e^{|Y_1|} \Bigg]\
 \Bigg (\sup_{k\in \Lambda }\sum_{\substack{Y \ni k}}|K(Y)| e^{|Y|} \Bigg)^{N-1},
\end{align*}
which proves the claim.
\end{proof}

\section{$H^{0|2}$ model and the arboreal gas}
\label{sec:arboreal gas}

In \cite{Caracciolo2004,Jacobsen2005,Caracciolo2007} it was proven, that the partition function of the arboreal gas can be represented as a Grassman integral. 
This integral is precisely the  partition function of the $H^{0|2m}$ model,   introduced
in Section \ref{sec:expansion}, with $m=1.$
In this section, we give an alternative proof for this duality based on the Hubbard--Stratonovich transformation and the matrix-tree theorem.

We start by reformulating the model in a sligthly more general setting. Let  $G= (\Lambda,E)$ be a finite
undirected graph with vertex set $\Lambda $
and edge set $E.$ The $H^{0|2}$ model on $G$ with edge weights 
$\beta =(\beta_{ij})_{\{i,j \}\in E}\in (0,\infty )^{E}$
and vertex weights $\varepsilon= \{\varepsilon_{i} \}_{i\in \Lambda }\in [0,\infty )^{\Lambda }$ has the partition function,
using  \eqref{eq:def-W-nu} and \eqref{eq:fact-notation} with $m=1,$
\[
 \tilde{Z}_{\beta,\veps}^{\Lambda}= \int \dd \nu_{\Lambda} \, \ee^{-\cw(\Lambda)} = \int \dd \psi_{\Lambda}
 \prod_{i\in \Lambda }\frac{\ee^{-\varepsilon_{j} (z_{i}-1)}}{z_{i}} \, \ee^{-\sum_{\{i,j\}\in E} \beta_{ij}(-1-\psi_i\cdot \psi_j+z_iz_j) }.
\]
We can recover formulas \eqref{eq:Gibbs-measure-Grassmann} and \eqref{eq:tilde-H} for $m=1$ by replacing
$\beta_{ij}=\beta J_{ij},$ fixing $\varepsilon_{j}=\varepsilon $
$\forall j\in \Lambda,$ and setting  $E$ to be the set of all unordered pairs $\{i,j \}$ such that $J_{ij}>0.$
We can  rewrite this partition function as a Grassmann Gaussian integral with a quartic perturbation.
This is the content of the next lemma.
\begin{lemma}
    \label{lemma:gaussian_perturbation} With the above notation,  it holds
    \begin{equation*}
     \tilde{Z}_{\beta,\veps}^{\Lambda}=  \int \dd \psi_{\Lambda}  \
\ee^{- \big(\cpsi, (-\Delta^{\beta}+1+\hat{\varepsilon}) \psi \big)} \ee^{-\sum_{\{i,j\}\in E}\beta_{ij}\cpsi_i \psi_i \cpsi_j \psi_j},
    \end{equation*}
 where, for any matrix $A\in \mathbb{R}^{\Lambda \times \Lambda}$ we define
 $ \left(\cpsi, A \psi \right)=\sum_{i,j\in \Lambda} \cpsi_i A_{i,j}\psi_j, $ 
 $\Delta^{\beta}$ is the graph Laplacian with edge weights $\beta_{ij},$
$1$ is the identity matrix and $\hat{\varepsilon}$ is the diagonal matrix with entry $\hat{\varepsilon}_{i,i}=\varepsilon_i$.
\end{lemma}
\begin{proof}
    Inserting the identities $ z_i-1=\cpsi_i\psi_i$ and $z_{i}^{-1}=\ee^{-\cpsi_i\psi_i}$
    we get
\[
       \frac{\ee^{-\varepsilon_{j} (z_{j}-1)}}{z_{j}}=\ee^{-(1+\varepsilon_i)\cpsi_i\psi_i}.
\]
Furthermore we have, using also $z_{i}z_{j}= (1+\cpsi_i\psi_i) (1+\cpsi_j\psi_j)=1+\cpsi_i\psi_i+\cpsi_j\psi_j+\cpsi_i\psi_i\cpsi_j\psi_j ,$
\[
 \beta_{ij}(-1-\psi_{i}\cdot \psi_{j}+z_iz_j)= \beta_{ij} (-\cpsi_i\psi_j-\cpsi_i\psi_j+\cpsi_i\psi_i+\cpsi_j\psi_j+\cpsi_i\psi_i\cpsi_j\psi_j).
\]
Summing over all edges we obtain
\begin{align*}
      \sum_{\{i,j\}\in E} \beta_{ij}(-1-\psi_i\cdot \psi_j+z_iz_j)  =(\cpsi,-\Delta^{\beta}\psi)+\sum_{\{i,j\}\in E}\cpsi_i\psi_i\cpsi_j\psi_j \; .\end{align*}  
Putting these identities together yields the claim.
\end{proof}
We introduce now the arboreal gas measure.
\begin{definition}
 Let  $G= (\Lambda,E)$ be a finite undirected graph with vertex set $\Lambda $
and edge set $E,$  $\beta =(\beta_{ij})_{\{i,j \}\in E}\in (0,\infty )^{E}$ a family of edge weights
and $\varepsilon= \{\varepsilon_{i} \}_{i\in \Lambda }\in [0,\infty )^{\Lambda }$ 
a family vertex weights.

Let $\mathcal{F} (G)$ be the set of all spanning forests on $G$. For each $F\in \cal{F} (G)$ we denote
by $E (F)$  the set of edges in the forest, and by $\mathcal{T} (F)$ the set of connected components (trees) in $F.$
Similarly, for each tree $T\in \mathcal{T} (F)$ we denote by $V (T)$ the set of vertices and by $E (T)$
the set of edges.
With these notations, 
    the probability of a forest $F\in \cal{F} (G)$ under the arboreal gas measure is
    \begin{equation*}\label{eq:def_arboreal_gas}
        \mathbb{P}_{\beta,\varepsilon}(F)=\frac{1}{Z_{\beta,\veps}^{\mathrm{arb},\Lambda}}\prod_{\{i,j\}\in E(F)}\beta_{ij}\prod_{T\in \mathcal{T} (F)}\left(1+\sum_{i\in V(T)}\varepsilon_i\right) \; ,
    \end{equation*}
 where the
    partition function is given by
    \begin{equation*}
        Z_{\beta,\veps}^{\mathrm{arb},\Lambda}=\sum_{F \in\mathcal{F} (G)}\prod_{\{i,j\}\in E(F)}\beta_{ij}\prod_{T\in \mathcal{T} (F)}\left(1+\sum_{i\in V(T)}\varepsilon_i\right) \; .
    \end{equation*}
  \end{definition}
The main result of this section is the following identity.
\begin{theorem}
    The partition function of the arboreal gas model and of the $H^{0|2}$ model  coincide:
    \begin{equation*}
        \tilde{Z}_{\beta,\veps}^{\Lambda}=Z_{\beta,\veps}^{\mathrm{arb},\Lambda} \; .
    \end{equation*}    
\end{theorem}    
\begin{proof}
 Using the representation from Lemma \ref{lemma:gaussian_perturbation}, we write
    \begin{equation*}
     \tilde{Z}_{\beta,\veps}^{\Lambda}=\int \dd \psi_{\Lambda}  \  
\ee^{- \big(\cpsi, (-\Delta^{\beta}+1+\hat{\varepsilon}) \psi \big)} \prod_{\{i,j\}\in E}
\ee^{-\beta_{ij}\cpsi_i \psi_i \cpsi_j \psi_j}.
    \end{equation*}
The quartic interaction can be reformulated as
    \[
        \cpsi_i \psi_i \cpsi_j \psi_j=\frac{(\cpsi_i \psi_i+\cpsi_j \psi_j)^2}{2} \; .
    \]
 We decouple these fourth order terms  by applying the Hubbard-Stratonovich transformation edgewise: for each edge $\{i,j\} \in E$ we introduce an auxiliary Gaussian field $\varphi_{ij}$ and write
\[
\ee^{ -  \beta_{ij}\cpsi_i \psi_i \cpsi_j \psi_j }= \ee^{-\frac{\beta_{ij}}{2}(\cpsi_i \psi_i+\cpsi_j \psi_j)^2}=
\int_{\mathbb{R}} \tfrac{\dd\varphi_{ij}}{\sqrt{2\pi }}\ \ee^{ -\frac{\varphi_{ij}^2}{2} }\, \ee^{-\ii \sqrt{\beta_{ij}}\varphi_{ij}(\cpsi_i\psi_i+\cpsi_j\psi_j)}
.
\]
Substituting this  representation into the Grassmann integral and exchanging the order of integration, we obtain
    \[
        \tilde{Z}_{\beta,\veps}^{\Lambda}=\int_{\mathbb{R}^{E}}\prod_{\{i,j\}\in E} \tfrac{\dd\varphi_{ij}}{\sqrt{2\pi }}\
\prod_{\{i,j\}\in E}  \ee^{ -\frac{\varphi_{ij}^2}{2} }\,
\Big[ \prod_{j \in \Lambda} \partial_{\bar{\psi}_{j}}\partial_{\psi_{j}}\Big] \
\ee^{- \big(\cpsi, (-\Delta^{\beta}+1+\hat{\varepsilon}+ \ii D ) \psi \big)}
 \; ,
    \]
where $D$ is the diagonal matrix with diagonal element $ D_{i}= D_{i,i}:=\sum_{j:\{i,j \}\in E }\sqrt{\beta_{ij}}\varphi_{ij}.$   The Grassmann integral yields 
    \begin{equation}
        \label{zeq}
        \tilde{Z}_{\beta,\veps}^{\Lambda}=
\int_{\mathbb{R}^{E}}\prod_{\{i,j\}\in E} \tfrac{\dd\varphi_{ij}}{\sqrt{2\pi }}\
\prod_{\{i,j\}\in E}  \ee^{ -\frac{\varphi_{ij}^2}{2} }\,
 \det \left( -\Delta^{\beta}+1+\hat{\varepsilon}+ \ii D\right).
    \end{equation}
By the matrix-tree theorem (see {\cite{ABDESSELAM200451}} for a simple proof and many references) 
the determinant can be rewritten as 
    \[
        \det  \left(-\Delta^{\beta}+1+\hat{\varepsilon}+ \ii D \right)=\sum_{F\in \mathcal{F} (G)}\sum_{R\in \mathcal{R} (F)}\
	\prod_{\{i,j\}\in E(F)}\beta_{ij}\ \prod_{r\in R}(1+\varepsilon_r+\ii D_r) \; ,
    \]
    where the set $R\in \mathcal{R} (F):=\otimes_{T\in \mathcal{T} (F)}V (T)$ chooses a root point $r\in V (T)$
    for each tree $T$ in the forest.
    We expand the product over the roots into terms that depend on the Gaussian fields and terms that are independent of it
    \[
        \prod_{r\in R}(1+\varepsilon_r+\ii D_r)=\sum_{X\subset R}\prod_{r\in R\setminus X }(1+\varepsilon_r)\ii^{|X|}\prod_{r\in X}D_r \; .
    \]    
    Substituting this expansion into \ref{zeq} yields
\[
  \tilde{Z}_{\beta,\veps}^{\Lambda}=
\hspace{-0,2cm}\sum_{F\in \mathcal{F} (G)}\sum_{R\in \mathcal{R} (F)}\ \prod_{\{i,j\}\in E(F)}\hspace{-0,1cm}\beta_{ij}\  \sum_{X\subset R}\prod_{r\in R\setminus X }(1+\varepsilon_r)\ii^{|X|} 
\int_{\mathbb{R}^{E}}\prod_{\{i,j\}\in E} \tfrac{\dd\varphi_{ij}}{\sqrt{2\pi }}\
\prod_{\{i,j\}\in E}  \ee^{ -\frac{\varphi_{ij}^2}{2} }\, \prod_{r\in X}D_r,
\]
hence we are left with evaluating the integral   over the $\varphi$ variables.
This can be reformulated,    unfolding the definition of $D_r$ and expanding the product over the roots, as 
    \begin{equation}\label{eq:prod_D_int}
       \sum_{\{ j_r: \{r,j_{r} \} \in E\}_{r\in X} }\prod_{r\in X}\sqrt{\beta_{j_r r}}
 \int_{\mathbb{R}^{E}}\prod_{\{i,j\}\in E} \tfrac{\dd\varphi_{ij}}{\sqrt{2\pi }}\
\prod_{\{i,j\}\in E}  \ee^{ -\frac{\varphi_{ij}^2}{2} }\, \prod_{r\in X}\varphi_{rj_{r}}
\;.
    \end{equation}
By symmetry of the Gaussian measure, the integral vanishes unless each variable $\varphi$ appears with an even power.
As each root $r\in X$ just contributes one Gaussian field $\varphi_{j_r r}$, this condition can only be satisfied if $|X|$ is even and
the set $X$ can be partitioned into pairs such that two vertices $r,r'$ in each pair correspond to the same edge   $\{r,j_{r} \}=\{r',j_{r'} \}.$
This can only happen if $r=j_{r'}$ and $r'=j_{r}.$
    Hence
\begin{align*}
     \eqref{eq:prod_D_int}&= \hspace{-0,4cm} \sum_{\{ j_r: \{r,j_{r} \} \in E\}_{r\in X} }\sum_{\Pi\; \mathrm{part}_{2}\; X} \prod_{\{r,r' \}\in \Pi}
     \ind{r=j_{r'}}\ind{r'=j_{r}} \ind{\{r,r' \}\in E} \, \beta_{rr'}\\
&\qquad =  \ind{|X| \; \mathrm{even}} 
\sum_{\Pi\; \mathrm{part}_{2}\; X}\  \prod_{\{r,r' \}\in \Pi}  \ind{\{r,r' \}\in E} \, \beta_{rr'} \;,
    \end{align*}
    where the sum is over the pair partitions $\Pi$ of $X$. 
Substituting this into the forest expansion gives
    \begin{equation*}
        \tilde{Z}_{\beta,\veps}^{\Lambda}=\hspace{-0,2cm}\sum_{F\in \mathcal{F} (G)}\sum_{R\in \mathcal{R} (F)}\\
	\sum_{\substack{X\subset R \\  |X| \mathrm{even}}}\sum_{\Pi\, \mathrm{part}_{2}\, X} \
\prod_{\{i,j\}\in E(F)}\hspace{-0,1cm}\beta_{ij} \prod_{\{r,r' \}\in \Pi}  \hspace{-0,1cm}\ind{\{r,r' \}\in E}  \beta_{rr'}
\prod_{r\in R \setminus X }(1+\varepsilon_r)
\hspace{-0,1cm}  \prod_{\{r,r' \}\in \Pi}  (-1)
 \;,
    \end{equation*}
  where $ \prod_{\{r,r' \}\in \Pi} (-1)=(-1)^{|X|/2}= \ii^{|X|},$  holds since $|X|$ is even.
    Let $\mathfrak{F}=F\cup \Pi$ be the forest obtained by adding the pairing edges to $F$.
A tree $T$ in this new forest $\mathfrak{F}$ is either a tree in the old one $F$ or the union of two trees in $F$ linked by a pair in $\Pi.$
In the first case the tree comes with a factor $ \sum_{r\in V (T)} (1+\varepsilon_{r}).$ In the second case we get factor $(-1)$ from all possible ways of removing 
a link from $T$ creating two disjoint trees. 
As a result we can write
  \begin{equation*}
       \tilde{Z}_{\beta,\veps}^{\Lambda}=\sum_{\mathfrak{F}\in \mathcal{F} (G)}\prod_{e \in E(\mathfrak{F})}\beta_{e} \prod_{T \in \mathfrak{F}} W(T)\; ,
    \end{equation*}
    where 
\[
W (T) = \sum_{r\in V (T)} (1+\varepsilon_{r})+  \sum_{\{i,j \}\in T} (-1)= |V (T) |+ \big ( \sum_{r\in V (T)} \varepsilon_{r} \big) - |E (T) |= 1+\sum_{r\in V (T)} \varepsilon_{r}.
\]

\begin{figure}[t]
\centering    
\begin{tikzpicture}[node distance=1cm]

\tikzset{
    treenode/.style={circle, draw, minimum size=1.2mm, inner sep=0pt},
    edge/.style={-,line width=0.15mm},
}

\node[treenode] (A) at (-0.2,0) {};
\fill[black] (-0.2,0) circle (1mm);
\node[above=2mm of A] {$r$};
\node[treenode] (B) at (1,0.9) {};
\node[treenode] (C) at (2,-1.1) {};
\node[treenode] (D) at (-2,0.8) {};
\node[treenode] (E) at (1.6,0) {};
\node[treenode] (F) at (0,-2) {};
\node[treenode] (G) at (2,-3) {};
\node[treenode] (H) at (-2,-3) {};
\node[treenode] (I) at (-2.2,-0.9) {};

\def\i{7}
\node[treenode] (A') at (-0.2+\i,0) {};
\fill[black] (-0.2+\i,0) circle (1mm);
\node[above=2mm of A'] {$r$};
\node[treenode] (B') at (1+\i,0.9) {};
\node[treenode] (C') at (2+\i,-1.1) {};
\node[treenode] (D') at (-2+\i,0.8) {};
\node[treenode] (E') at (1.6+\i,0) {};
\node[treenode] (F') at (0+\i,-2) {};
\fill[black] (0+\i,-2) circle (1mm);
\node[below=2mm of F'] {$r'$};
\node[treenode] (G') at (2+\i,-3) {};
\node[treenode] (H') at (-2+\i,-3) {};
\node[treenode] (I') at (-2.2+\i,-0.9) {};

\node (X) at (0.55,-0.2) {$(1+\veps_{r})$};
\node (X') at (-1+\i,-1.3) {$-\beta_{rr'}$};

\draw[edge] (A) -- (B);
\draw[edge] (A) -- (I);
\draw[edge] (A) -- (F);
\draw[edge] (A) -- (D);
\draw[edge] (C) -- (E);
\draw[edge] (C) -- (F);
\draw[edge] (G) -- (F);
\draw[edge] (H) -- (F);

\draw[edge] (A') -- (B');
\draw[edge] (A') -- (I');
\draw[line width=0.9mm] (A') -- (F');
\draw[edge] (A') -- (D');
\draw[edge] (C') -- (E');
\draw[edge] (C') -- (F');
\draw[edge] (G') -- (F');
\draw[edge] (H') -- (F');

\def\j{-4}
\node (T) at (0,\j) {Tree weight: $(1+\veps_{r}) \prod_{\{l,l' \}\in E(T)} \beta_{ll'}$};
\node (T') at (\i,\j) {Tree weight: $-\beta_{rr'} \prod_{\{l,l' \}\in E(T)\setminus \{ r,r'\}} \beta_{ll'}$};
\end{tikzpicture}
\end{figure}

\end{proof}

\paragraph{Acknowledgements} We are grateful to Tyler Helmuth for useful discussions.
This work has been supported by the German Research Foundation (DFG) under Germany's Excellence Strategy - GZ 2047/1,
Project-ID 390685813 and by the European Union’s Horizon 2020
research and innovation programme under the Marie Sk\l{}odowska-Curie grant agreement No 101154394.

 {\footnotesize

}
\end{document}